\documentclass[12pt]{article}\usepackage[]{graphicx}\usepackage[]{xcolor}
\makeatletter
\def\maxwidth{ %
  \ifdim\Gin@nat@width>\linewidth
    \linewidth
  \else
    \Gin@nat@width
  \fi
}
\makeatother

\definecolor{fgcolor}{rgb}{0.345, 0.345, 0.345}

\usepackage{framed}
\makeatletter
 {\par\unskip\endMakeFramed%
 \at@end@of@kframe}
\makeatother

\definecolor{shadecolor}{rgb}{.97, .97, .97}
\definecolor{messagecolor}{rgb}{0, 0, 0}
\definecolor{warningcolor}{rgb}{1, 0, 1}
\definecolor{errorcolor}{rgb}{1, 0, 0}
\newenvironment{knitrout}{}{} % an empty environment to be redefined in TeX

\usepackage{alltt}
\usepackage{amsmath,amssymb,amsthm,graphicx,natbib,geometry,hyperref}
\usepackage{authblk}

\newcommand{\Exp}[2][]{{\rm E}_{#1}[ #2 ]  }
\newcommand{\Var}[1]{{\rm Var}[ \ensuremath{ #1 } ]  }
\newcommand{\Cov}[1]{{\rm Cov}[ \ensuremath{ #1 } ]  }

\newcommand{\bl}[1]{{\mathbf #1}}
\newcommand{\bs}[1]{\boldsymbol #1}
\newcommand{\tr}{\text{tr}}

\newcommand{\ubar}[1]{\underbar{$#1$}} 
\newcommand{\minrank}{\hat r} 
\newcommand{\maxrank}{\check r} 
\newcommand{\xrank}{r}

\newtheorem{theorem}{Theorem}
\newtheorem{corollary}{Corollary}
\newtheorem{definition}{Definition}
\newtheorem{lemma}{Lemma}

\IfFileExists{upquote.sty}{\usepackage{upquote}}{}
\begin{document}

\title{Extended rank regression for all ordinal data} 

\author{Peter Hoff and Supratik Basu} 
\affil{Department of Statistical Science, Duke University} 

\maketitle
\begin{abstract} 
The accuracy of inference from a regression model depends largely 
on how well the model represents the relationship between the mean and variance 
of the outcomes. 
As this relationship is rarely of direct interest, 
it is natural to treat 
it as a nuisance parameter, rather than attempt to estimate it. 
We take this approach in the context of 
 a monotonically transformed linear regression model using  a pseudo-likelihood based on 
an extended notion of ranks. 
This approach can accommodate a wide range of mean-variance relationships and 
 any ordinal data type, including continuous and discrete ordered data, and
 requires no estimation or prior specification of the transformation, or decision to treat 
an outcome as continuous or discrete. 
 We show that the extended rank likelihood incurs no asymptotic information loss
at the two extremes of continuous  and binary data, and that
rank-based prediction intervals can obtain approximate coverage control 
conditional on the features. 
Bayesian parameter estimates and prediction intervals are available via 
a simple Gibbs sampling  algorithm. For settings where the model is in doubt, 
conformal calibration of the Bayesian predictive distribution
provides intervals  with guaranteed marginal frequentist coverage.

\smallskip
\noindent \textit{Keywords:} 
Box--Cox transformation, conditional coverage, conformal prediction, marginal likelihood, ordinal data, posterior prediction, rank likelihood, semiparametric regression, variance-stabilizing transformation. 

\end{abstract}

\section{Introduction}
Primary uses of regression models include inferring 
the relationship between outcome variables $Y_1,\ldots, Y_n$ 
and their corresponding 
feature vectors $\bl x_1,\ldots, \bl x_n \subset  \mathbb R^p$, and
prediction of a new outcome $Y$ given a new feature vector $\bl x$. 
Typically, 
it is assumed that 
$Y_1,\ldots, Y_{n}$ are conditionally independent given 
$\bl x_1,\ldots, \bl x_{n}$, and that the conditional distribution of 
$Y_i$ given $\bl x_1,\ldots, \bl x_{n}$ depends only 
on $\bl x_i$. The most widely-used such model is the normal linear regression model, 
which posits that $Y_i \sim N( \bl x_i^\top \bs\beta, \sigma^2)$ independently for $i=1,\ldots, n$. 
A critical yet often violated assumption of this model is that the variance 
is the same for all observations. If this is untrue, then the ordinary least squares estimate 
of $\bs\beta$ remains unbiased, but is no longer variance-optimal, as it would generally 
have a larger variance matrix (in Loewner order) than an appropriately-weighted 
least squares estimate. 
Additionally, confidence intervals for $\bs\beta$ and prediction intervals 
for new outcomes can be misleading if constructed using an inappropriate assumption of 
constant variance.

\begin{figure} 
\begin{knitrout}
\definecolor{shadecolor}{rgb}{0.969, 0.969, 0.969}\color{fgcolor}
\includegraphics[width=\maxwidth]{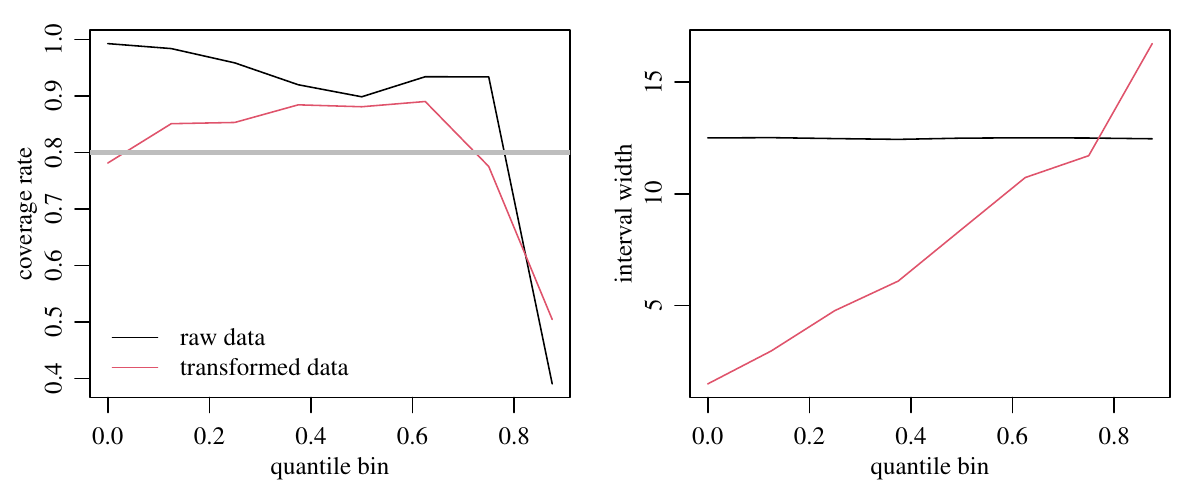} 
\end{knitrout}
\caption{Seattle rain data. The left panel plots conditional coverage rates of 
nominal 80\% 
prediction intervals for daily rainfall,  
and the right panel plots their widths. }
\label{fig:sr0} 
\end{figure}

To illustrate this phenomenon, consider forecasting daily rainfall $Y_{t+1}$ 
in Seattle from one-day lagged weather 
features $\bl x_{t+1}$, over a ten-year period from 2016--2025 (these data are more fully described in Section 5). From a linear model fit to observed data $(y_1,\bl x_1),\ldots, (y_t, \bl x_{t} )$, 
an 80\% prediction interval can be constructed for each $Y_{t+1}$ such that, if the normal 
linear model is accurate, the coverage probability of the interval will be 
80\%, conditionally on each $\bl x_{t+1}$. However, as shown in the left panel of 
Figure \ref{fig:sr0}, this coverage rate is not maintained empirically: 
While the marginal coverage rate across all predicted outcomes is 88\%, 
the rate is above 90\% 
for outcomes below the 87.5 percentile of rainfall values and below 40\% for the 228 
outcomes above the 87.5 percentile. The variable coverage rate is due primarily to the fact that 
the interval widths are largely determined by a shared estimate of $\sigma^2$, 
whereas the data strongly indicate that the variance is increasing with the mean. 
An estimate of this mean-variance relationship suggests that fitting a linear model 
to  $(y_1^{1/4},\bl x_1),\ldots, (y_t^{1/4}, \bl x_{t} )$, 
 would be more appropriate \citep{box_cox_1964}. 
As shown in the left panel of the figure, constructing prediction intervals from the linear model fit 
on this quarter-power scale, then transforming back to the original scale, 
yields intervals whose coverage rate is somewhat closer to the target rate across different levels 
of rainfall. This improved conditional coverage is partly because  
the estimated variance on the original data scale
is now increasing as a function of the mean, resulting in interval widths that scale more appropriately 
with the magnitude of the outcome, as shown in the right panel of the figure. 

Such power transformations can help stabilize the variance
of many types of positive data,
leading to more accurate inference. 
More generally, one may consider the generalized Box--Cox transformation model, 
which posits that each $Y_i$ is equal to $G(Z_i)$ where $Z_i \sim N(\bl x_i^\top \bs\beta, 1)$ 
and  $G$ is an unknown increasing function.  
This function partly determines the 
 mean-variance relationship of the outcome variables, which depending on $G$ could be 
constant, increasing, decreasing and 
even non-monotonic. 
One approach to inference for such models is to jointly estimate $G$ along with $\bs\beta$, 
either parametrically \citep{box_cox_1964,bickel_doksum_1981}, or nonparametrically 
\citep{horowitz_1996,ye_duan_1997,
chen_2002,
hothorn_kneib_buhlmann_2014}. These works considered the case of  
strictly increasing  $G$ and hence continuous outcomes. 
For discrete ordinal outcomes  the standard approach is to use ordered probit models 
\citep{mckelvey_zavoina_1975,mccullagh_1980,albert_chib_1993,kowal_wu_2025} that parameterize 
$G$ as an increasing step function with discontinuities at each point in the sample space of 
$Y_1,\ldots, Y_n$. Such models are primarily used for ordered categorical 
data, where the sample space consists of a finite set of ordered categories that are known 
in advance, and may not correspond to specific numerical outcomes. 
While these models could  in principle be applied to numerical outcomes such as the Seattle rain data, 
doing so can be somewhat awkward and inefficient, as a large fraction of the ``categories'' in such 
cases may have zero  observations. 
For example, the Seattle rain data is measured to the nearest 1/100th millimeter, and would thus require 
a parameter space for $G$ of dimension more than 5300 just to accommodate the range of smallest (0.00) to largest (53.11) 
rainfall values, even though only 1097 unique values appear in the dataset. 

Although the transformation $G$ plays a role in determining the conditional distributions of the outcomes, 
 it is often not itself of direct interest. This point was raised in different ways 
by \citet{box_cox_1964,box_cox_1982} and \citet{bickel_doksum_1981}, the latter also highlighting 
statistical challenges in joint estimation of $\bs\beta$ and $G$. 
This suggests the use of statistical methods that treat $G$ as a nuisance parameter.
In the case of strictly increasing $G$, 
\citet{pettitt_1982} suggested estimating $\bs\beta$ by approximately maximizing 
the rank likelihood, 
a type of  marginal likelihood that depends on the data only though the ranks, and hence is a function only 
of $\bs\beta$. 
Asymptotic properties of this maximum rank likelihood estimator  have been 
studied in the case of strictly increasing $G$ by \citet{bickel_1987} and \citet{bickel_ritov_1997}, with the 
latter showing that for semiparametric estimation of $\bs\beta$
there is no asymptotic information loss in reducing the data to the ranks. 

While the different types of monotonic transformation models mentioned above 
are widely used, the available methodology 
for their implementation tends to be tailored to specific cases --- such as positive, 
continuous or discrete data --- and parametric or nonparametric specifications of $G$. 
Rather than selecting from this variety of implementations, a data analyst might find it
useful to have 
a single generic methodology that can be applied to all ordinal data types, 
particularly those where it is unclear whether to treat the outcome as discrete or continuous. 
To this end, in this article we review and develop inference and prediction methods for generic ordinal data, using a monotonically transformed linear model (MTLM) as described above where the only assumption 
on $G$ is that it is non-decreasing. This is accomplished primarily via Bayesian computational 
tools applied to 
a pseudo-likelihood based on an extended notion of ranks, 
first used by \citet{pettitt_1984} in an application to ordinal categorical data. 
In the next section, we review this extended rank likelihood (ERL) and describe how Bayesian inference 
for the regression coefficients $\bs\beta$ may be obtained from a very simple Gibbs sampler. 
In Section 3, we study the efficiency and asymptotic properties of rank-based inference. 
We complement existing 
results of \citet{bickel_ritov_1997} for the case of continuous outcomes (strictly increasing $G$) 
with a new result for binary outcomes that indicates that there is no asymptotic efficiency loss 
in estimating $\bs\beta$ via the ERL, as compared to using a full likelihood that requires specification of $G$. 
Section 4 covers rank-based prediction: First 
we show how a prediction interval constructed from a consistent rank-based estimate can achieve
a conditional coverage rate that asymptotically  matches its nominal level. 
We then develop two practical methods of rank-based prediction interval construction. 
First, we show how a posterior predictive distribution for the extended rank of $Y$ among 
$Y_1,\ldots, Y_n$ may be obtained from the ERL. This predictive distribution may be 
combined with the observed values of $Y_1,\ldots, Y_n$ to yield 
a posterior prediction interval for the value of $Y$. For scenarios where modeling 
assumptions are in doubt, we also provide a conformal prediction procedure based 
on the posterior predictive distribution of the extended rank of $Y$, which 
can guarantee a target marginal coverage rate even if the model is misspecified. 
Two example data analyses are presented in Section 5, and a discussion follows in Section 6. 
Mathematical proofs are in the appendix. 

\smallskip

Replication code for all numerical examples in this article, as well as an open-source {\sf R}-package {\tt perle}, are available at the first author's website.

\section{Rank-based inference for the MTLM} 

\subsection{Rank and extended rank likelihood}
The monotonically transformed linear model (MTLM) 
for a vector of scalar outcomes $\bs Y= (Y_1,\ldots, Y_n)\in \mathbb R^n$  and a given design matrix $\bl X \in \mathbb R^{n\times p}$ 
specifies that 
there exist a $\bs \beta \in \mathbb R^p$ and a non-decreasing function $G:\mathbb R \rightarrow \mathbb R$ 
such that 
\begin{align}  
 (Z_1,\ldots, Z_n) = \bs Z & \sim  N_n(\bl X\bs \beta,\bl I_n)  \label{eqn:lmod} \\
Y_i & = G(Z_i ), \label{eqn:tmod}
\end{align}  
with both $\bs\beta$ and $G$ being unknown. 
Note that a model that specified the variance of $\bs Z$ as $\sigma^2  \bl I_n$ 
for some unknown $\sigma^2>0$ would be non-identifiable, as would a model where $\bl X$ included an intercept term. 

If $\bs Z$ were observed then estimation and inference for 
$\bs \beta$ could proceed without regard for the unknown and 
possibly infinite-dimensional 
parameter $G$. 
 \citet{pettitt_1982} recognized that while $\bs Z$ is not observed, 
some information about $\bs Z$ is available from the observed 
data vector $\bs Y$ 
that does not 
rely on knowledge of $G$ other than its monotonicity. In particular, 
observation of $Y_i < Y_{i'}$ implies 
$Z_i< Z_{i'}$. 
If $G$ is strictly increasing then 
there are no ties among the  
elements of $\bs Y$ and 
the rank ordering of $\bs Z$ is the same as that of
 $\bs Y$. 
For this case, \citet{pettitt_1982} suggested making inference for $\bs \beta$ 
using the \emph{rank likelihood} $L( \bs \beta : \bl r)$, defined as 
$L(\bs\beta : \bl r ) = \Pr( r(\bs Y) = \bl r | \bs\beta)  =
 \Pr( r(\bs Z) = \bl r | \bs\beta)$ where $\bl r = r(\bl y)$ are the 
ranks of the observed data vector $\bl y$. In particular, he proposed estimating $\bs\beta$ 
with the value $\hat{\bs\beta}$ that maximizes an integral approximation to $L(\bs\beta:\bl r)$. 

The rank likelihood was extended by \citet{pettitt_1984} 
to accommodate 
ordinal categorical data. In the context of multivariate 
copula estimation, \citet{hoff_2007a} pointed out how a multivariate version 
of this extension
can be applied to the general setting 
where $G$ is non-decreasing but not necessarily strictly increasing. 
In this case, a range of $z$-values could  map to the same $y$-value, resulting 
in a distribution with atoms. As such, allowing 
$G$ to be non-decreasing  instead of strictly increasing results in a model that accommodates 
data types that are continuous, discrete, or some combination of these. 
In the context of the MTLM, this extended rank likelihood is 
defined by 
\begin{equation}
 L( \bs\beta: S(\bl y) ) = \Pr( \bs Z \in S(\bl y) | \bs\beta ) \label{eqn:erl}, 
\end{equation} 
where 
$S(\bl y)$ is the convex set of possible $\bs Z$-values implied by the monotonicity of $G$: 
\begin{equation}
S(\bl y) = \{ \bl z \in \mathbb R^n :  
     \max\{ z_{i'} : y_{i'} < y_i \} <    z_i < 
     \min\{ z_{i'} : y_{i} < y_{i'} \}  \}. 
\end{equation} 
Note that in the absence of ties, $L(\bs\beta: S(\bl y) ) = L(\bs\beta: \bl r)$. 
Whether or not there are ties, $L(\bs\beta: S(\bl y) )$ 
does not depend on the unknown value of $G$, suggesting its use 
as a marginal likelihood for semiparametric inference for $\bs\beta$. 

In the case of non-atomic data the ranks have been well-studied as a statistic, 
that is, a known function of the data. 
For data that possibly include ties, the event 
$\bs Z \in  S(\bl y)$ can be represented in terms of a generalization of the rank statistic, which we refer to as 
the \emph{extended ranks}:

\begin{definition}
For  $\bl y\in R^m$, define the  \emph{minimum ranks} $\minrank(\bl y)$  and the 
\emph{maximum ranks} $\maxrank(\bl y)$ 
as $m$-dimensional vectors of integers with elements $(\minrank(\bl y)_1,\ldots,\minrank(\bl y)_m )$ 
and $(\maxrank(\bl y)_1,\ldots,\maxrank(\bl y)_m )$ 
given by 
\begin{align} 
\minrank( \bl y)_i &  = 1+ \sum_{i'=1}^m 1\times (y_{i'} < y_i )  \\
\maxrank( \bl y)_i &  =\sum_{i'=1}^m 1\times (y_{i'} \leq y_i ). 
\end{align} 
Define the \emph{extended ranks} of $\bl y$ as an $m$-tuple of sequences of integers, with the $i$th sequence given by
\begin{equation}
 \xrank( \bl y )_i = \{ j\in \mathbb N:   \minrank( \bl y)_i  \leq j \leq \maxrank( \bl y)_i \}. 
\end{equation}
\end{definition}
Note that $\xrank(\bl y)_i = \xrank(\bl y)_{i'}$ if and only if $y_i = y_{i'}$, and that the extended ranks 
may be ordered 
 so that
 $\xrank(\bl y)_i < \xrank(\bl y)_{i'}$ if $y_{i}< y_{i'}$.
Additionally, if there are no ties among the elements of a vector $\bl y$, then the minimum and maximum ranks 
are equal to each other and to the usual rank,  
and so in this case $r(\bl y)$ reduces to the usual definition of the ranks of $\bl y$. 

A connection between the extended ranks and monotone transformation models can be made as follows:
\begin{lemma} 
\label{lem:xrankG} 
Let $\bl z\in \mathbb R^n$ be a vector with no ties, and $y_i = G(z_i)$ where $G$ is non-decreasing. 
Then for each $i=1,\ldots, n$,
$r(\bl z )_i \in  \xrank( \bl y)_i$. 
\end{lemma} 
In other words, the extended ranks of $\bl y$ determine the possible ranks of $\bl z$ implied by the monotonicity of $G$. 
This lemma will be used  in Section 4 where we show how 
a prediction region for the rank of a new $z$-value 
implies a prediction region for 
the extended ranks of the corresponding $y$-value. 
The lemma also provides a representation of the extended rank likelihood in terms of the 
extended ranks of $\bl y$: 
\begin{align*}
  L(\bs\beta : S(\bl y) ) & \equiv \Pr( \bs Z \in S(\bl y) | \bs\beta ) \\ 
 & = \Pr( r(\bs Z )_i \in  \xrank( \bl y)_i, i=1,\ldots ,n  | \bs\beta ). 
\end{align*}
We note that this likelihood is not equivalent to the probability of observing 
$r(\bs Y) =r(\bl y)$ unless $G$ is strictly increasing, because for discrete data 
the distribution of ties 
among the elements of $\bs Y$ will generally depend on $G$. 
See \citet{hoff_2007a} for a discussion of this point in the context of 
rank-based copula estimation.

\subsection{Posterior approximation} 
In principle, estimation and inference for $\bs\beta$ in the non-atomic case could proceed via maximization and differentiation of 
 $L(\bs\beta: \bl r)$. To this end, 
\citet{pettitt_1982} provided 
a deterministic integral approximation to $L(\bs \beta:\bl r)$ whose validity 
depends on $\bs\beta$ being close to zero. For other situations, 
\citet{doksum_1987} and \citet{pettitt_1987} provided Monte Carlo 
approximation methods, mostly based on iteratively reweighted least-squares approximations where the moments at each iteration are obtained via Monte Carlo approximation. The latter article compared parameter estimates 
given by these approximations in a simulation study, but indicated 
that using these approximations
for inference beyond point estimation 
is difficult. 

In contrast, 
\citet{hoff_2008b} 
pointed out that
Bayesian inference for $\bs\beta$ based on either the rank or  extended rank likelihood 
is easy to obtain 
using a very simple Gibbs sampling algorithm. 
Consider a normal prior distribution for $\bs\beta$
with density $\pi(\bs\beta)$. 
Given the information that $\bs Z \in S(\bl y)$, the uncertainty about $\bs\beta$ 
is then  described by
the conditional  density  $\pi(\bs\beta |  \bs Z \in S(\bl y)) \propto  \pi(\bs\beta) \times L(\bs\beta: S(\bl y))$. 
This posterior density is the 
$\bs\beta$-marginal density of the 
joint conditional density $\pi( \bs\beta, \bs Z | \bs Z\in S(\bl y))$. 
As such, a MCMC approximation of the former can 
be obtained from the $\bs\beta$-values of a MCMC approximation to the latter. 
For the particular prior distribution $\bs\beta\sim N_p(0,\tau^2I)$, the steps of a Gibbs sampler to 
approximate the distribution with density $\pi( \bs\beta, \bs Z | \bs Z\in S(\bl y))$ are as follows:
Given a current state $\bs Z = (Z_1,\ldots, Z_n)$ and $\bs\beta$, 
\begin{enumerate}
\item iteratively for $i\in \{1,\ldots,n\}$, 
 \begin{enumerate}
\item compute the interval $s(\bl y,\bs Z_{-i})  = (\max\{Z_{i'}:y_{i'}<y_i\} , \min\{Z_{i'}:y_{i}<y_{i'}\})$, 
\item simulate $Z_i \sim N( \bl x_i^\top \bs\beta , 1)$ constrained to $s(\bl y, \bs Z_{-i})$. 
\end{enumerate}
\item simulate $\bs\beta\sim N_p(  \bl m , \bl V )$, where 
   \begin{itemize}
  \item   $\bl m =  (\bl X^\top \bl X + \bl I_p/\tau^2)^{-1} \bl X^\top \bs Z$ ; 
  \item   $\bl V =  (\bl X^\top \bl X + \bl I_p/\tau^2)^{-1}$.
\end{itemize}
\end{enumerate}
Iteration of this algorithm generates a Markov chain 
with stationary distribution equal to the 
distribution of $(\bs Z, \bs\beta)$ conditional on $\bs Z \in S(\bl y)$. 
The empirical distribution of the simulated $\bs\beta$-values may be used as a 
Monte Carlo approximation to 
the distribution with density  $\pi(\bs\beta |  \bs Z \in S(\bl y))$. 
We refer to estimation and inference using this posterior density as ``posterior extended rank likelihood estimation'' (PERLE).

Note that in the presence of ties among the observed outcomes, 
$z$-values corresponding to the same $y$-value are constrained to the same interval, and so the speed of the algorithm can be considerably increased by updating these values 
simultaneously. 
Additionally, mixing and convergence of the Markov chain can be 
improved by making generalized 
Gibbs updates of the form $(\bs Z,\bs\beta) \mapsto (c \bs Z, c\bs\beta)$ for an 
appropriately 
simulated random scalar $c$ \citep{liu_sabatti_2000}:
If $\bs Z \in S(\bl y)$ then $c \bs Z\in S(\bl y)$ for any 
$c>0$, and so if $\bs\psi = (\bs Z, \bs\beta)$ is a valid state of the 
Markov chain then so is $c \bs\psi$. For cases such as these, 
Liu and Sabatti suggest a generalized Gibbs update of $\bs\psi$ 
with target stationary density $\pi$ 
by 
simulating a value $c$ from the density proportional to 
$\pi ( c  \bs\psi ) \times c^{n+p-1}$ and then updating the state to $c \bs\psi$. 
This provides a simultaneous update to $\bs Z$ and $\bs\beta$ that reduces autocorrelation in the Markov chain while maintaining the target stationary distribution. 
For the MTLM, the appropriate distribution from which to simulate $c$ is 
such that $c^2 \sim \text{gamma}((n+p)/2 , ( ||\bs Z - \bl X\bs\beta||^2 +  ||\bs\beta||^2/\tau^2 )/2)$. 
Further computational details are available from {\tt perle}, the companion {\sf R}-package to this article.

\section{Rank-based estimation}

As described in the Introduction, the ERL is a marginal likelihood that depends
on $\bs{\beta}$ but not on the unknown transformation $G$, and so may not be as
informative as a full likelihood that depends on both $\bs{\beta}$ and $G$. The
potential information loss can be quantified via the score function and observed
information matrix of the ERL, which we derive below. We then show that in the
two extreme cases of ordinal data, continuous outcomes at one extreme   and binary
outcomes at the other, the ERL incurs no asymptotic efficiency loss.

For a given set $S = S(\bl{y})$, the log ERL is
$l(\bs{\beta}: S) = \log \Pr(\bs{Z} \in S | \bs{\beta})$. 
Differentiating with respect to $\bs{\beta}$, the score function is
$\dot{l} = \dot{P}/P$ where $P = \Pr(\bs{Z} \in S | \bs{\beta})$. Letting
$f(\bl{z}|\bs{\beta})$ denote the density of $\bs{Z}$ at $\bl{z} \in \mathbb{R}^n$,
we have
\begin{align*}
\dot{P} = \frac{d}{d\bs{\beta}} \int_S f(\bl{z}|\bs{\beta})\, d\bl{z}
        = \int_S \dot{l}(\bs{\beta}:\bl{z})\, f(\bl{z}|\bs{\beta})\, d\bl{z},
\end{align*}
where $\dot{l}(\bs{\beta}:\bl{z}) = d\log f(\bl{z}|\bs{\beta})/d\bs{\beta}$ is the
score function based on observation of $\bs{Z} = \bl{z}$. Dividing by $P$ gives
$\dot{l}(\bs{\beta}:S) = \Exp{\dot{l}(\bs{\beta}:\bs{Z}) |  \bs{Z} \in S}$, 
so the score of the ERL is the conditional expectation of the complete-data score
given $\bs{Z} \in S$. Since $\bs{Z} \sim N_n(\bl{X}\bs{\beta}, \bl{I}_n)$, we have
$\dot{l}(\bs{\beta}:\bl{z}) = \bl{X}^\top(\bl{z} - \bl{X}\bs{\beta})$, which gives
\begin{equation}
\dot{l}(\bs{\beta}:S) = \bl{X}^\top\left(\Exp{\bs{Z}| \bs{Z}\in S} -
\bl{X}\bs{\beta}\right).
\end{equation} 
The second derivative matrix of $l$ is $\ddot l = \ddot  P/P - \dot P \dot P^\top/P^2$, 
where 
\begin{equation} 
 \ddot P =  \int_S  \left (  \ddot l(\bs\beta:\bl z) + \dot l(\bs\beta:\bl z) \dot l(\bs\beta:\bl z)^\top  \right ) 
  f(\bl z | \bs \beta)  \, d\bl z. 
\end{equation} 
Plugging in the values of $\ddot l(\bs\beta: \bl z)$ and $\dot l(\bs\beta: \bl z)$ gives the observed 
ERL information matrix $I_{ERL}(\bs\beta)$: 
\begin{align}
I_{ERL}(\bs \beta) =  - \ddot l (\bs\beta: S)  & =  
   \bl X^\top \bl X - \bl X^\top \left ( \Var{ \bs Z |\bs Z\in S }\right ) \bl X  \nonumber \\
   & = 
   \bl X^\top ( \bl I_n - \Var{ \bs Z | \bs Z\in S } ) \bl X \preceq \bl X^\top \bl X.  
\label{eqn:infoloss} 
\end{align}   
where ``$\preceq$'' denotes the Loewner order.
That this information matrix is positive semidefinite
results from the following lemma:
\begin{lemma} 
\label{lem:ccineq}
Let $\bs Z\sim N_n( \bs\mu  , \bl I_n)$ for any $\bs\mu \in \mathbb R^n$ and let $S \subset \mathbb R^n$ be 
convex. 
Then  $\Var{ \bs Z | \bs Z\in S} \preceq \Var{\bs Z} = \bl I_n$. 
\end{lemma} 
Equation (\ref{eqn:infoloss}) shows how  $\Var{\bs Z | \bs Z\in S}$ quantifies the information loss from 
only observing $\bs Z\in S$ rather than observing $\bs Z$ completely. 
Intuitively, we expect that the coarser the partial ordering given by $S$, the higher 
the conditional variance and 
the greater the information loss. While it is difficult to fully investigate all degrees of coarseness, 
we can study the efficiency of estimators based on the ERL 
in the cases of minimal and maximal coarseness, that is, the cases
where $G$ is strictly monotone and there are no ties among the $Y_i$'s, and that where 
$G$ is a step function and the $Y_i$'s are binary.

In case of strictly increasing $G$, the location of the $Z_i$'s relative to each other is increasingly 
revealed as $n\rightarrow \infty$, although their absolute location is not. As such, 
we might expect the ERL to have similar asymptotic information as the
 normal linear regression model $\bs Y\sim N_n(\bl 1_n \alpha + \bs X\bs \beta , \bl I_n)$, 
that is, the case that $G$ is an unknown location shift.  
Recall that for this submodel of the MTLM  the efficient information for the MLE  
$\tilde {\bs \beta}$ is $I_{\rm eff}(\bs\beta) = \bl X^\top (\bl I_n -\bl 1 \bl 1^\top /n) \bl X$.  
For the special value $\bs\beta=\bl 0$, this intuition is correct:
\begin{theorem} 
\label{thm:effZero} 
Let $\bs\beta=\bl 0$ and $\bl x_1,\ldots, \bl x_n$  be i.i.d.\  with $\Exp{\|\bl x_i\|^2} < \infty$. 
Then as $n\rightarrow \infty$, 
\[ ( I_{\rm eff}(\bl 0) - I_{ERL}(\bl 0 ))/n =   \bl X^\top ( \Var{\bs Z | \bs Z \in S}  - \bl 1\bl 1^\top /n)\bl X /n  \stackrel{p}{\rightarrow} 0.  \]
\end{theorem}
More generally, \citet{bickel_ritov_1997} studied  
the asymptotic efficiency of the maximum rank likelihood estimator (MRLE) $\hat{\bs\beta}$. 
In particular, they asked the question:
among estimators of $\bs{\beta}$ that use $\bs{Y}$ but make no
assumptions about $G$ beyond strict monotonicity, is there any efficiency loss from
using the ERL rather than a full likelihood? The following theorem, due to
Bickel and Ritov (1997), says that there is no such loss. 

\begin{theorem}[Bickel and Ritov, 1997]\label{thm:info-bickel} 
Under some regularity conditions, 
for the MTLM with strictly increasing $G$
there exists a parametric submodel 
such that
$\sqrt{n}(\hat{\bs{\beta}} - \bs{\beta}) \overset{d}{\to}
N(\bl{0}, I_{\mathrm{eff}}^{-1}(\bs{\beta}))$, where 
$I_{\mathrm{eff}}(\bs{\beta})$ is the efficient information for $\bs{\beta}$ in the submodel. 
\end{theorem}
In other words, the MRLE is semiparametrically efficient: no estimator that is
consistent and asymptotically normal uniformly over all strictly monotone transformation
submodels --- even those based on full observation of $\bs Y$ --- can have asymptotic variance smaller than
$I_{\mathrm{eff}}^{-1}(\bs{\beta})$, the asymptotic variance of the MRLE $\hat{\bs\beta}$. 

The opposite extreme of information loss occurs in the case  of  
binary $Y_i$'s where the transformation $G$ reduces to a step function
with a single unknown threshold $\alpha \in \mathbb{R}$. 
Letting 
$N_0 = \{ i: y_i = 0\}$ be the indices corresponding to 
the zero outcomes, 
the set $S(\bl y)$ may be written as the union of $|N_0|$ disjoint sets $S_1(\bl y),\ldots, S_{|N_0|}(\bl y)$
of the form $S_k (\bl y)= \{ \bl z\in \mathbb R^n :  \max \{z_i: i \in N_0\setminus\{k\}\} < z_k < 
 \min\{z_i:  i\not\in N_0\}\}$.
The ERL may therefore be written as 
\begin{align*} 
L(\bs\beta: S(\bl y)) & =  \sum_{k\in N_0} \Pr( \bs Z \in S_k(\bl y)|\bs\beta ) \\
& = \sum_{k\in N_0}  \int 
 \left\{
\prod_{i\neq k} 
  \Phi(\alpha-\bl{x}_i^{\top}\bs{\beta})^{1-y_i} (1-\Phi(\alpha-\bl{x}_i^\top\bs{\beta}))^{y_i } \right \} \phi(\alpha-\bl{x}_k^\top\bs{\beta})\,d\alpha \\ 
& = 
 \int 
 \left\{
\prod_{i=1}^n
  \Phi(\alpha-\bl{x}_i^{\top}\bs{\beta})^{1-y_i} (1-\Phi(\alpha-\bl{x}_i^\top\bs{\beta}))^{y_i } \right \}\left  (\sum_{k\in N_0} \frac{  \phi(\alpha-\bl{x}_k^\top\bs{\beta})}{\Phi(\alpha-\bl x_k^\top \bs\beta )} \right )   \, d\alpha
\\ 
& \equiv \int L( \bs\theta : \bl y) w(\bs\theta ) \, d\alpha
\end{align*} 
where $\bs\theta= (\alpha,\bs\beta)$, and $L(\bs\theta:\bl y)$ is the usual probit regression likelihood function 
for $\bs\theta$ given the data vector $\bl y$. 
The ERL in this case 
resembles an integrated likelihood for $\bs{\beta}$ obtained by treating
$\alpha$ as a nuisance parameter with a pseudo-prior density $\pi(\alpha|\bs\beta) \propto w (\bs\theta)$. 
As such, 
Bayes-type estimators, such as the PERLE 
 described in Section 2.2, 
 are  asymptotically efficient  under standard conditions:
\begin{theorem}\label{thm:rmle-sqrt-n}  
For each $n$ let $(\bl x_1,Y_1) ,\ldots, (\bl x_n,Y_n)$ be i.i.d.\  
with $\Pr(Y_i=0|\bl x_i) = \Phi(\alpha - \bl x_i^\top \bs\beta)= 1-\Pr(Y_i=1|\bl x_i)$, 
  $\Exp{\|\bl x_i\|^6} < \infty$,
 and $\Var{\bl  x_i}$  strictly positive definite. 
Let $\tilde{\bs\theta} =  (\tilde \alpha, \tilde{\bs\beta})$ be the MLE  of $\bs\theta=(\alpha,\bs\beta)$
 based on the full probit likelihood,  
and let $\hat{\bs\beta}$ be the expectation of $\bs\beta$ under the probability distribution with 
density proportional to $\pi_\beta (\bs\beta ) L(\bs\beta:S(\bs Y))$ where $\pi_\beta(\bs\beta)$ is a nonsingular multivariate normal density. 
Then 
\begin{align*} 
\sqrt{n}(\hat{\bs\beta} - \tilde{ \bs \beta}) & = o_p(1)  \\
\sqrt{n}(\hat{\bs \beta} - \bs \beta) & \stackrel{d}{\rightarrow} N_p(\bl{0},\, [I(\bs\theta)^{-1}]_{\beta\beta}),
\end{align*} 
where  $[I(\bs\theta)^{-1}]_{\beta\beta}$ is the $(\bs\beta,\bs\beta)$ block of $I(\bs\theta)^{-1}$
and the asymptotic variance of  $\tilde{\bs\beta}$. 
\end{theorem}
Based on the proof of the theorem, we also expect the result to hold for any log-concave prior density
$\pi_\beta$ such that $\|\nabla \log  \pi_\beta \|$ is locally bounded. 

Taken together, Theorems \ref{thm:info-bickel} and \ref{thm:rmle-sqrt-n} 
show that rank-based estimation via the
ERL incurs no asymptotic efficiency loss at the two extremes of the ordinal data
spectrum.
While we do not have specific results for intermediate
cases, we expect that, with considerable additional bookkeeping, the proof of
Theorem~\ref{thm:rmle-sqrt-n} can be extended to general ordered probit models
with $K$ categories for any finite $K$, and that  the PERLE or 
other estimators based on the ERL 
will achieve asymptotic
efficiency across the full range of ordinal data types considered in this article.

\section{Rank-based prediction} 

\subsection{Coverage of rank-based prediction intervals}
Before describing two practical methods of prediction interval construction, we first 
provide some general asymptotic results regarding rank-based prediction intervals $C(\bs Y)$ 
for a new observation $Y=G(Z)$, $Z\sim N(\bl x^\top\bs\beta,1)$ using 
data $\bs Y$ from the model $Y_i=G(Z_i)$, $\bs Z \sim N_n(\bl X\bs\beta,\bl I_n)$. 
We show how a rank-based interval procedure based on a $\sqrt{n}$-consistent estimate 
$\hat{\bs\beta}$ of $\bs\beta$, such as those described in the previous section, can provide 
approximately constant conditional coverage, that is, $C(\bs Y)$ 
satisfies  $\Pr ( Y  \in C( \bs Y) | \bs\beta ) \approx  1-\alpha$ 
for all $\bl x\in \mathbb R^p$.
While the interval procedure studied here differs somewhat from 
the PERLE procedures described in the next subsections, 
the results presented here indicate that rank-based prediction intervals can 
provide asymptotically constant conditional coverage.

If $\bs\beta$ were known, then 
 integers $l$ and $u$ could be selected 
such that 
$ \Pr( Z_{(l)}  <  Z<  Z_{(u)}  |\bs \beta )\approx 1-\alpha$, 
where $Z_{(k)}$ is the $k$th order statistic of $Z_1,\ldots, Z_n$ and 
 dependence of the probability on 
$\bl x$ and $\bl X$ is suppressed for notational simplicity.  
In the case of strictly increasing $G$ where 
the ranks of $\bs Z$ and $\bs Y$ are identical, it follows 
that
$ \Pr( Y_{(l)} < Y< Y_{(u)}  | \bs\beta)  =
 \Pr( Z_{(l)}  <  Z<  Z_{(u)}  |\bs \beta )\approx 1-\alpha$, 
and so in this way, 
a prediction interval for the rank of $Z$ among $\bs Z$ provides a prediction interval 
$(Y_{(l)}, Y_{(u)})$ for $Y$. 
In the case of ties, 
an interval with greater than $1-\alpha$ coverage could be constructed by appropriately 
expanding the interval, as will be done in the next two subsections.

As $\bs\beta$ is not known, 
neither is the distribution of the rank of $Z$. 
However, we expect that
a plug-in estimate of the rank distribution using a sufficiently good 
estimate $\hat{\bs\beta}$ of $\bs\beta$ 
will provide intervals with approximate $1-\alpha$ coverage. 
To make this more rigorous, for $\gamma\in (0,1)$ and $\bl b \in \mathbb R^p$, let $H(\gamma,\bl b ) = \Pr( Z < Z_{(\lfloor \gamma n \rfloor ) } | \bl b)$, where the probability is evaluated 
under $\bs Z\sim N_n( \bl X \bl b, \bl I_n)$ independently of $Z \sim N(\bl x^\top \bl b, 1)$. 
As a function of $\gamma$, $H$ is essentially the cumulative distribution function (CDF) of 
the rank of $Z$ among $Z_1,\ldots, Z_n$. In particular, if $\gamma_l$ and $\gamma_u$ 
satisfy $H(\gamma_u,\bs\beta) - H(\gamma_l, \bs\beta ) = 1-\alpha$ then 
the sequence of integers  from $\lfloor \gamma_l n \rfloor$ to  $\lceil \gamma_u n \rceil$ 
provides a prediction interval for the rank of $Z$ 
with at least $1-\alpha$ coverage. 

The accuracy of a plug-in approximation to $H(\gamma, \bs\beta)$ can be quantified 
as follows:
\begin{lemma}  
\label{lem:lipschitz} 
For any $\gamma\in (0,1)$ and $\bs\beta, \tilde {\bs\beta} \in \mathbb R^p$, 
\[ | H(\gamma, \bs\beta ) - H(\gamma ,\tilde {\bs\beta}) | \leq (2\pi)^{-1/2} 
\left (
   ||\bl X( \bs\beta-\tilde{\bs\beta})||_{\infty} + |\bl x^\top (\bs\beta-\tilde{\bs\beta})|  
\right ). 
\]
\end{lemma}
The result follows from a simple coupling argument and some Lipschitz-type  inequalities. 
From this, we can assess the asymptotic accuracy of 
 $H(\gamma, \hat{\bs\beta})$ as  an estimate of $H(\gamma, \bs\beta)$:
\begin{corollary} 
\label{cor:rootnH}
Suppose $\sqrt{n} (\hat{\bs\beta } - \bs\beta) = O_p(1)$ and $\max\{ 
||\bl x_1||, \ldots, ||\bl x_n||\}  = O(1)$. Then  
\[
\sup_{\gamma} | H(\gamma, \bs\beta ) - H(\gamma ,\hat {\bs\beta}) |  = O_p(n^{-1/2}). 
\] 
\end{corollary}  

As mentioned above, if $\bs\beta$ were known then a prediction region for the rank of $Z$ could be obtained from 
$\gamma_l,\gamma_u$ that satisfy
$H(\gamma_u,\bs\beta) - H(\gamma_l, \bs\beta ) \approx 1-\alpha$.  
Absent knowledge of $\bs\beta$, we instead 
obtain estimates $\hat \gamma_l,\hat\gamma_u$ that 
satisfy $H(\hat \gamma_u,\hat{\bs\beta}) - H(\hat \gamma_l, \hat{\bs\beta} ) \approx 1-\alpha$. 
The coverage of this interval evaluated under the true $\bs\beta$ will converge to the target rate:
\begin{theorem} 
\label{thm:hconv}
Under the assumptions of Corollary \ref{cor:rootnH}, suppose $\hat \gamma_l,\hat\gamma_u$ 
satisfy $H(\hat\gamma_u, \hat{\bs\beta}) - H(\hat\gamma_l, \hat{\bs\beta})  
\stackrel{p}{\rightarrow } 1-\alpha$ as $n\rightarrow \infty$. 
Then $H(\hat\gamma_u, \bs\beta) - H(\hat\gamma_l, \bs\beta) 
\stackrel{p}{\rightarrow } 1-\alpha$ as $n\rightarrow \infty$. 
\end{theorem}
The procedures in the next subsections essentially construct prediction 
regions from values of $\hat \gamma_l, \hat\gamma_u$  that are measurable with respect 
to the ranks of $Z_1,\ldots, Z_n$. 
If additionally the design is random, in that $\bl x_1,\ldots, \bl x_n \sim$ i.i.d.\ $P_x$, 
(so that marginally, $Z_1,\ldots, Z_n$ are i.i.d.), then 
the frequentist coverage rate of $(Z_{(\lfloor \hat \gamma_l n \rfloor )},
 Z_{(\lceil \hat \gamma_u n \rceil )})$ as a prediction interval for $Z$ is 
 $\Exp{ H(\hat\gamma_u, \bs\beta) - H(\hat\gamma_l, \bs \beta) }$. Under the above assumptions, 
this converges to $1-\alpha$ for any fixed feature $\bl x$ for the new value of $Z$: 
\begin{theorem}   
\label{thm:cconv} 
Let $\hat\gamma_l$ and $\hat\gamma_u$ be functions of the ranks of $Z_1,\ldots, Z_n$.  
Then under the assumptions of Theorem \ref{thm:hconv}, 
$\Pr( Z_{(\lfloor \hat \gamma_l n \rfloor )} < Z < 
 Z_{(\lceil \hat \gamma_u n \rceil )} | \bs\beta , \bl x) \rightarrow 1-\alpha$, 
where the probability is calculated with respect to 
 $Z_i | \bl x_i \sim N(\bl x_i^\top \bs\beta,1), \bl x_i \sim P_x$ independently for 
$i=1,\ldots, n$ and independent of $Z\sim N(\bl x^\top \bs\beta,1)$ for any fixed $\bl x\in \mathbb R^p$. 
\end{theorem} 
We note that the assumption on $\bl x_1,\ldots, \bl x_n$  can be relaxed. 
For example, it is sufficient  that
$\max\{ 
||\bl x_1||, \ldots, ||\bl x_n||\}  = o(n^{1/2})$, although the convergence of the coverage probability will be at a slower rate.

Finally, in the case that $G$ is strictly increasing so that the ranks of $\bs Z$ are the same as 
those of $\bs Y$, the events 
$ Y_{(\lfloor \hat \gamma_l n \rfloor )} < Y < 
 Y_{(\lceil \hat \gamma_u n \rceil )} $ and
$ Z_{(\lfloor \hat \gamma_l n \rfloor )} < Z < 
 Z_{(\lceil \hat \gamma_u n \rceil )} $ 
are the same,  and so 
 $( Y_{(\lfloor \hat \gamma_l n \rfloor )},Y_{(\lceil \hat \gamma_u n \rceil )}) $ 
  is a prediction interval for $Y$ with asymptotic 
coverage $1-\alpha$, marginally over $\bs Y$ and $\bl X$. 
Specific methods for $G$ not strictly increasing are discussed in the next subsection.

\subsection{Bayesian posterior prediction with extended ranks} 
In the previous subsection we approximated the rank distribution  $\Pr( Z < Z_{(k)}  | \bs\beta)$ 
with a plug-in estimate based on a consistent estimate  $\hat{\bs\beta}$ of $\bs\beta$. 
The Bayesian analog of the plug-in estimate
 is the posterior predictive rank distribution,  $\Pr( Z < Z_{(k)}  | \bs Z \in S(\bl y))$, which 
integrates over the uncertainty in $\bs\beta$. We first discuss how this distribution can be computed
 and used to 
form a prediction interval for the rank of $Z$,  then describe how this interval generates 
a prediction interval for $Y=G(Z)$. 

Let $Z \sim N( \bl x^\top \bs\beta,1)$ be independent of $\bs Z \sim N_n( \bl X \bs\beta , \bl I_n)$. 
The posterior predictive distribution of the rank $r(\bs Z_{n+1})_{n+1}$ of 
$Z$ among $\bs Z_{n+1} = ( Z_1,\ldots,Z_n, Z)$ 
conditional on $\bs Z\in S(\bl y)$
may be approximated by adding the following two steps to each 
iteration of the Gibbs sampler 
described in Section 2.2:
\begin{enumerate}
\item Sort the $n$ values of $\bs Z$ to form $-\infty = Z_{(0)}< Z_{(1)} < \cdots < Z_{(n)} 
 < Z_{(n+1)} = \infty$;  
\item For each  $k=1,\ldots, n+1$ compute 
\[ \Pr( r( \bs Z_{n+1} )_{n+1}   = k | \bs Z,\bs\beta) = 
                             \Phi( Z_{(k)} - \bl x^\top \bs\beta ) -
                             \Phi( Z_{(k-1)} - \bl x^\top \bs\beta ), \]
\end{enumerate}
where $\Phi$ is the standard normal CDF. 
Averaged over the iterations of the Markov chain, the probabilities calculated above 
provide a Monte Carlo approximation to 
$\Pr( r(  \bs Z_{n+1} )_{n+1}   = k | \bs Z \in S(\bl y))$ for each possible rank  $k=1,\ldots, n+1$ of $Z$. 
From this, a subsequence 
 $R = \{ l+1,\ldots, u\}$  of $\{1,\ldots, n+1\}$ 
can be identified such that 
 $\Pr( r( \bs Z_{n+1} )_{n+1}    \in R | \bs Z \in S(\bl y)) 
\equiv  \Pr( Z_{(l)} < Z <  Z_{(u) }  | \bs Z \in S(\bl y)) 
\geq 1-\alpha$, so 
 $R$ has at least $1-\alpha$ posterior coverage of the rank of $Z$, conditional on $\bs Z \in S(\bl y)$. 

The set $R$ can be used to make predictions about the extended rank of $Y$ among the observed 
values $\bl y$ of $\bs Y$, 
and thus about the value of $Y$. 
To build intuition we first consider the simpler case that $G$ is strictly increasing, so that
the ranks of $\bs Y_{n+1} = (Y_1,\ldots, Y_n,Y)$ are the same as those of $\bs Z_{n+1} = (Z_1,\ldots,Z_n, Z)$. 
In this case, 
to the Bayesian with information $\bs Z \in S(\bl y)$ 
\[ 
\Pr( Y_{(l)}< Y < Y_{(u) }    | \bs Z \in S(\bl y) ) =  \Pr( Z_{(l)}< Z < Z_{(u) }    | \bs Z \in S(\bl y) ) 
  \geq 1-\alpha,  
\]
and so $C(\bs Y) = [Y_{(l)}, Y_{(u)} ]$ 
provides a prediction procedure for $Y$ 
that this Bayesian assesses as having at least $1-\alpha$ coverage. 

Generalizing this prediction interval to the case that $G$ is increasing but not strictly increasing
requires some additional book-keeping
because in this case the 
rank of $Z$  might not be the same as the rank of $Y$. 
However, 
the rank of $Z$  does determine the  possible values for the extended rank of 
$Y$, from which an interval for $Y$ can be constructed using the numerical values 
of $\bl y$. 
Let 
$y_{(1)} < \cdots < y_{(K)}$ be the ordered unique values of  the observed outcome vector $\bl y$, 
with multiplicities $n_1,\ldots, n_K$, so that $\sum_k  n_k = n$. 
Further, let $s_1,\ldots, s_K$ be the cumulative multiplicities, so that 
  $s_1=n_1$, $s_k = \sum_{j=1}^k n_j$ and $s_K = n$. Then
\begin{enumerate}
\item  if $r(\bs Z_{n+1})_{n+1} = s_k+1$ then $y_{(k)}    \leq   Y \leq  y_{(k+1) }$;
\item if  $s_{k}+1<r(\bs Z_{n+1})_{n+1} < s_{k+1} + 1$ then $ Y = y_{(k+1)}$, 
\end{enumerate}  
where for completeness, we set 
 $s_0=0$, 
$s_{K+1}= n+1$,   
and $y_{(0)}$ and $y_{(K+1)}$ to be the smallest and largest 
possible $y$-values.
The first item holds because the condition on the rank implies that $Z$ is above 
all $Z_i$'s  for which $G(Z_i) = y_{(k)} $ and below 
all $Z_i$'s  for which $G(Z_i) = y_{(k+1)} $. Absent knowledge of $G$, 
$Y$ could be equal to $y_{(k)}$, equal to $y_{(k+1)}$ or in-between.
The second item holds because the condition on the rank implies that $Z$ is between the smallest 
and largest of the $Z_i$'s for which $G(Z_i) = y_{(k+1)}$, and so $Y$ must also equal $y_{(k+1)}$. 

From the above relationships 
it is straightforward to show that if $Z_{(l)} < Z < Z_{(u)}$, 
that is, the rank of $Z$ is between $l+1$ and $u$ inclusive, then  
tight lower and upper bounds $\ubar y,\bar y$ for $Y$ are given by 
\begin{itemize}
\item $\ubar y = y_{(k_l)}$ where $k_l= \min\{ k: l\leq s_k \}$; 
\item $\bar y =  y_{(k_u)}$  where $k_u= \min\{ k: u\leq s_k \}$. 
\end{itemize}
Thus if  $l$ and $u$ satisfy $\Pr( Z_{(l)} < Z <  Z_{(u) }  | \bs Z \in S(\bl y)) \geq 1-\alpha$, 
then a $1-\alpha$ prediction interval for $Y$ is given by $C(\bs Y) = [\ubar y,\bar y]$. 
A more intuitive formula for the interval can be expressed in terms of an extended quantile function 
based on the empirical CDF $\hat F$ of the observed data vector $\bl y$. Define
$\hat F^{-1}(p)$ as 
\[ 
\hat F^{-1}(p) = \left \{ \begin{array}{ll} 
  y_{(0)} & \text{ for $p=0$, }\\ 
  \min \{ y: \hat F(y) \geq  p\} & \text{ for $p\in (0,1]$, }\\
  y_{(K+1)} & \text{ for $p>1$. }
 \end{array} \right . 
\] 
Then the PERLE prediction interval for $Y$ 
is $C(\bs Y) = [\hat F^{-1}(l/n), \hat F^{-1}(u/n)]$. 

This interval procedure for $Y$   is not a posterior prediction interval in the usual sense,  
as the posterior distribution of $r(\bs Z_{n+1})_{n+1}$ is computed only
conditional on the event $\bs Z \in S(\bl y)$, and not conditional on 
having observed the specific values of $y_{(1)},\ldots, y_{(K)}$ or their multiplicities $n_1,\ldots, n_K$.
In particular, the Bayesian who knows  only 
that $\bs Z$ lies in some set $S$, 
but not  how $S$ is obtained,
cannot construct 
this interval. However, such a Bayesian would evaluate the procedure as having  $1-\alpha$ 
coverage, marginally over the values of $Y_1,\ldots, Y_n$: 
The event $r(\bs Z_{n+1})_{n+1}\in \{  l+1\ldots, u\} $ is a subset of the event $Y \in  [\ubar{Y} ,\bar Y ]$ 
where $\ubar{Y}$ and $\bar Y$ are determined by the construction described above applied to the 
values of $\bs Y$. 
 Therefore, to the Bayesian 
who only conditions on  
 $\bs Z \in S$, 
\begin{align*} 
\Pr(  \ubar Y\leq Y \leq \bar Y   | \bs Z \in S )
 \geq \Pr( Z_{(l)} <  Z < Z_{(u)} | \bs Z \in S ) \geq 1-\alpha. 
\end{align*}
From this perspective
the PERLE interval has at least $1-\alpha$ coverage of $Y$, 
marginally over the values of $Y_1,\ldots, Y_n$.

\subsection{Conformal prediction with extended ranks} 
\label{sec:cpred} 

The Bayesian prediction interval constructed in the previous subsection 
relies on the appropriateness of the MTLM.  
If this is in doubt, then 
the MTLM can instead be used as a device to construct a 
nonparametric rank-based conformal prediction interval, for which 
exchangeability of $\{ (Y_1,\bl x_1),\ldots, (Y_{n},\bl x_{n}), (Y,\bl x) \}$ is sufficient to 
guarantee a marginal coverage rate of $1-\alpha$. 
As will be shown in Section 5, in an example with a strong  mean-variance relationship
the proposed rank-based interval can also approximately maintain
a target coverage rate conditional on  $\bl x$, whereas 
procedures using standard conformity scores will only achieve the target marginally. 

At a high level, the procedure we propose 
consists of the following steps:
\begin{enumerate}
\item[0.] Use a training dataset 
to fit the MTLM as described in Section 2.2 and obtain an approximation 
to the posterior distribution of $\bs\beta$. 
\item[1.]
Use calibration data $\{ ( Y_1,{\bl x}_1),\ldots, 
(  Y_{n},{\bl x}_n)\}$ 
to obtain a conformal prediction set for the extended rank of 
  $Y$ among $Y_1,\ldots, Y_n$, using a conformity score based on the predicted rank
of $Z$ among $Z_1,\ldots, Z_n$. 
\item[2.] Convert the prediction set for the extended rank of $Y$ to a prediction set for the 
value of $Y$ based on the observed outcomes of the calibration set. 
\end{enumerate}

Step 1 consists of a conformal algorithm that predicts $r_{n+1} \equiv r(\bs Y_{n+1})_{n+1}$, 
the extended rank of $Y$ among $Y_1,\ldots,Y_n$, 
from $r(\bs Y)$, the extended ranks of $Y_1,\ldots, Y_n$.
From $r(\bs Y)$ the values of $0=s_0 < s_1 < \cdots < s_K < s_{K+1} = n+1$ defined in Section 4.2 can be obtained, which determine the possible  extended ranks of $Y$. There are $2K+1$ possible values of $r_{n+1}$, 
corresponding to $Y$ 
 being less than all values of $\bs Y$, in-between pairs of consecutive values, equal to individual values, and 
being greater than all values. Letting $y_{(1)} < \cdots <  y_{(K)}$ be the observed 
ordered unique values of $\bs Y$, 
and defining $y_{(0)}$ and $y_{(K+1)}$ as the smallest and largest possible $y$-values, we have
\begin{align} 
 r_{n+1} & = 
\{ s_k+1 \} \text{ if }  y_{(k)} < Y < y_{(k+1)}    \label{eqn:ery1} \\
r_{n+1} & =
\{ s_{k}+1,\ldots,s_{k+1}+1\} \text{ if } Y =  y_{(k+1)}.   \label{eqn:ery2} 
\end{align} 

We evaluate each candidate extended rank $r_{n+1} \in \{ r_{n+1}^{1} ,\ldots, r_{n+1}^{2K+1} \}$ 
with the conformity score
\[ 
 s( \bl r)_{n+1} = 
\max_{j \in r_{n+1} } 
\Pr( r( \bs Z_{n+1})_{n+1} = j | \bl x_1,\ldots, \bl x_n, \bl x), 
\]
where 
 $\bl r$ is the list with elements equal to $r(\bs Y_{n+1})$, for which the last element is 
the candidate value $r_{n+1}$. 
We compute this  score with a 
Monte Carlo approximation given a distribution over $\bs\beta$-values, such as the posterior 
distribution obtained by fitting the MTLM to a training dataset  in Step 0. 
The motivation for this score is 
that it should result in a prediction region that mimics the Bayesian prediction region 
described in the previous subsection, which is constructed by including ranks for $Z$ 
that have high posterior probabilities.
As shown in \citet{hoff_2023}, conformity scores based on posterior predictive probabilities 
generally result in Bayes-optimal prediction regions, in terms of precision. 

A candidate extended rank $r_{n+1}$ for $Y$ is included in the conformal prediction set if
$s( \bl r)_{n+1}$ is
large compared to the corresponding conformity scores of the 
observed extended ranks 
$r_1,\ldots, r_n$, which are the first $n$ elements of $\bl r$. 
To assess this, for each $i=1,\ldots, n$ we compute 
the observed score
\[
 s( \bl r)_i = \max_{j\in r_i} \Pr( r( \bs Z_{n+1})_i =j | \bl x_1,\ldots, \bl x_{n},\bl x). 
\]
Computing this score for each $i=1,\ldots, n$ requires knowing the extended ranks 
of $\bs Y_{n+1}$, 
which seems to require knowledge of $Y$. However, 
the extended ranks
of $\bs Y_{n+1}$ 
can be determined from the observed extended ranks of $\bs Y$ and the 
candidate value $r_{n+1}$ of  $r(\bs Y_{n+1})_{n+1}$ as follows:
Changing notation slightly, 
let  $\xrank(y_1,\ldots, y_n) = (r_1^n,\ldots r_n^n))$, 
 and  $\xrank(y_1,\ldots,y_n, y) = (r_1^{n+1},\ldots , r_{n+1}^{n+1})$, 
with $r_i^{n} = \{ j\in \mathbb N :  \minrank_i^n \leq j \leq  \maxrank_i^n \}$
and  $r_i^{n+1} = \{ j\in \mathbb N :  \minrank_i^{n+1} \leq j \leq  \maxrank_i^{n+1} \}$. 
Then the minimum and maximum ranks of $(y_1,\ldots, y_n)$ and $(y_1,\ldots,y_n,y)$ are related  by 
\begin{itemize}
\item $\hat r_{i}^{n+1} = \hat r_{i}^{n}  + 1 \times ( \hat r_{i}^n \geq \check r_{n+1}^{n+1} )$
\item $\check r_{i}^{n+1} = \check r_{i}^{n}  + 1 \times (\check  r_i \geq \hat r_{n+1}^{n+1} ).$ 
\end{itemize}

A $1-\alpha$ prediction set $R\subset\{ r_{n+1}^1,\ldots, r_{n+1}^{2K+1}\}$  
for the extended rank of $Y$ may be constructed by including in $R$ each 
candidate 
 extended rank  $r_{n+1}$ for which $s_{n+1}(r_{n+1})$ is greater than or equal to 
the $\alpha$ quantile of $\{ s(\bl r)_1,\ldots, s(\bl r)_{n+1}\}$, that is, 
\[
  R =\left   \{ r_{n+1} \in \{ r_{n+1}^1,\ldots, r_{n+1}^{2K+1}\} :  
  \sum_{i=1}^{n+1} 1\times (s(\bl r)_{n+1}   \geq  s(\bl r)_i  ) \geq  \alpha \times (n+1)  \right  \}. 
\]
Note that the extended ranks $\bl r= r(\bs Y_{n+1})$ and calibration scores $s(\bl r)$ 
are recomputed for each candidate value $r_{n+1}$ of the extended rank of $Y$. 

As will be described in a moment, the conformal prediction set $R$ for the extended 
rank $r_{n+1}$ satisfies 
$\Pr( r( \bs Y_{n+1})_{n+1} \in R ) \geq 1-\alpha$, where the probability  
holds over exchangeable outcome-feature pairs 
 $\{ (Y_1,\bl x_1),\ldots, (Y_{n},\bl x_{n}), (Y,\bl x) \}$. 
From the set $R$, a prediction interval $C$ for $Y$ may be constructed using 
the observed unique values  of $\bs Y$ and Equations \ref{eqn:ery1} and \ref{eqn:ery2}. 
Let $l$  and $u$ be the indices of the smallest and largest 
extended ranks in $R$, respectively, and let $\ubar y$ be the lower bound implied by 
$r_{n+1}^{l}$ and $\bar y$ be the upper bound implied by $r_{n+1}^{u}$. Then 
the event $r(\bs Y_{n+1})_{n+1} \in R$ implies the event that 
$Y \in C \equiv [ \ubar y , \bar y ]$, hence the probability of the latter is larger than 
the former and so  $\Pr(Y \in [ \ubar y , \bar y ] ) \geq 1-\alpha$. 

Like typical conformal prediction procedures, 
if $\{ (Y_1,\bl x_1 ), \ldots, (Y_{n},\bl x_{n}), (Y, \bl x)\}$ are exchangeable then 
$\Pr( Y  \in  [ \ubar y , \bar y ]) \geq 1-\alpha$ whether or not the 
MTLM holds, where in this probability calculation 
the values $(\ubar y, \bar y)$ are random variables constructed from 
$\{ (Y_1,\bl x_1 ), \ldots, (Y_{n},\bl x_{n})\}$  and $\bl x$. 
This result follows from the exchangeability of the scoring function used to construct the 
prediction set: For $\bl y_{n+1}\in \mathbb R^{n+1}$, $\bl X_{n+1} \in \mathbb R^{(n+1)\times p}$, and 
a permutation 
$\pi$ of $\{ 1,\ldots, n+1\}$, write  $\pi\circ \bl y_{n+1}$ as the elements of $\bl y_{n+1}$ permuted by $\pi$
and
$\pi \circ \bl X_{n+1}$ as the matrix obtained by permuting the rows of $\bl X_{n+1}$ by $\pi$. 
For all such $\bl y_{n+1}$, $\bl X_{n+1}$ and $\pi$, the scoring function $s$, which we 
now write as a function of $\bl y_{n+1}$ and $\bl X_{n+1}$, satisfies
\begin{align*} 
s( r(\pi\circ \bl y_{n+1}) , \pi\circ \bl X_{n+1}) & = s( \pi \circ r(\bl y_{n+1}) , \pi\circ \bl X_{n+1})  \\
 & =  \pi \circ s( r(\bl y_{n+1}), \bl X_{n+1}),    
\end{align*}
where $\pi \circ r(\bl y_{n+1}) $ is the list of extended ranks of $\bl y_{n+1}$ permuted by $\pi$. 
Letting $\bs Y_{n+1} = (Y_1,\ldots, Y_n, Y)$, the above result and the exchangeability assumption imply
\begin{align*} 
s( r(\bs Y_{n+1} ) , \bl X_{n+1}) & \stackrel{d}{=}
s( r(\pi \circ \bs Y_{n+1} ) , \pi \circ \bl X_{n+1})  \\ 
 & =  \pi \circ s( r(\bs Y_{n+1}), \bl X_{n+1}), 
\end{align*} 
and so the vector of conformity scores is also exchangeable. 
From here, the usual argument for marginal coverage of conformal prediction procedures applies: 
The probability that $r(\bs Y_{n+1})_{n+1}$ is in the set $R$ described above is equal to the 
probability that the score for $r(\bs Y_{n+1})_{n+1}$ is greater than the $\alpha$-quantile of 
the elements of $s( r(\bs Y_{n+1}), \bl X_{n+1})$, which by exchangeability is greater than or equal to $1-\alpha$. 
Finally, because
the event $r( \bs Y_{n+1}  )_{n+1} \in R$ implies the event 
$Y \in C$, we have 
\[ 
 \Pr( Y \in C )  \geq  
  \Pr( r( \bs Y_{n+1}  )_{n+1} \in R ) \geq 1-\alpha, 
\]
and so the random interval  $C=[ \ubar y , \bar y ]$ has at least $1-\alpha$ coverage for 
$Y$, marginally over values of $\bs Y_{n+1}$ and $\bl X_{n+1}$. We refer to this interval procedure 
as the PERLE-conformal prediction interval.

\section{Examples}

\subsection{Seattle rainfall}

We  model and predict daily rainfall in Seattle over a ten-year 
period from 2016 through 2025, using data from the 
\href{https://power.larc.nasa.gov/}{NASA POWER project}.  
Figure \ref{fig:sr1} displays the first two years of data as a time series and a histogram. 
The data exhibit strong seasonality, with a rainy season running October through March
and dry season running July through September. The data are also highly skewed and somewhat 
discrete, being recorded to the nearest 1/100th of a millimeter. 
Zero rain was recorded on 436 days, and the total number of unique recorded values of 
rainfall was $K=1097$ out of $n=3652$ 
total recorded values. 

We build a predictive model for rainfall as a function of 29  
features, including 
one-day lag values of rainfall, humidity and windspeed at locations 
northwest, due west and southwest of Seattle (nine features); 
sine and cosine functions to capture seasonality in rainfall (two features);
and interactions between the nine lagged weather variables and the sine and cosine 
functions (eighteen features).  
A normal-scores transformation was applied to each of the nine 
weather variables before the interaction terms were constructed. 
Inclusion of these types of features is fairly typical for empirical models of local 
rainfall \citep{vogel_etal_2020}. 
An initial analysis using a quarter power transformation indicated a small degree of  
residual autocorrelation (lag-1 and lag-2 sample autocorrelations of 0.056
  and -0.071, 
respectively). 

\begin{figure} 
\begin{knitrout}
\definecolor{shadecolor}{rgb}{0.969, 0.969, 0.969}\color{fgcolor}
\includegraphics[width=\maxwidth]{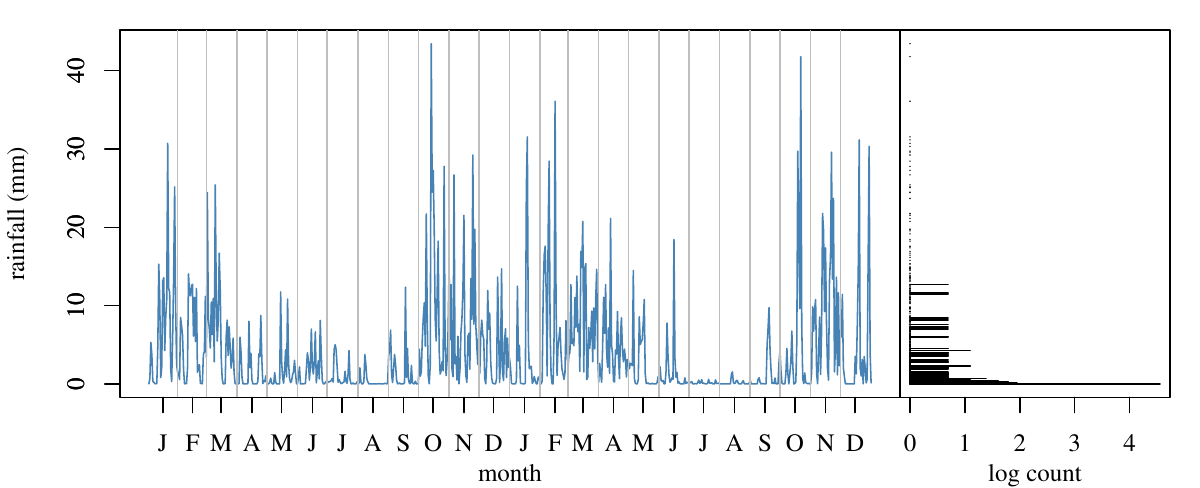} 
\end{knitrout}
\caption{Seattle rainfall data from 2016-01-01 through 2017-12-31. The left panel 
displays the data as a time series, the right panel as a histogram. } 
\label{fig:sr1}
\end{figure}

Using all ten years of data ($n=3652$), we ran the Gibbs sampler described in Section 2.2 for 11,000 
iterations, dropping the first 1000 iterations to allow for convergence of the Markov chain 
to the stationary distribution. Parameter values at every 10th iteration were retained, 
resulting in 1000 $\bs\beta$ values with which to approximate the posterior 
distribution. 
Mixing of the Markov chain was very good:  
across all parameter values
the smallest effective sample size based on the MCMC sample of size 1000 
was 892.
Posterior mean estimates (PERLEs) and 95\% posterior quantile intervals for the regression parameters are shown in the left panel of Figure 
\ref{fig:sr2}. The right panel compares $t$-scores from an OLS fit using the quarter-power transformed 
data to analogous $t$-scores from the posterior distribution of $\bs\beta$, obtained 
by dividing the posterior mean of each coefficient by its posterior standard deviation. 
The plot indicates that, in terms of assessing significance of the features, the two approaches are 
nearly equivalent, even though the PERLE $t$-scores do not require a pre-specified data transformation.

\begin{figure}
\begin{knitrout}
\definecolor{shadecolor}{rgb}{0.969, 0.969, 0.969}\color{fgcolor}
\includegraphics[width=\maxwidth]{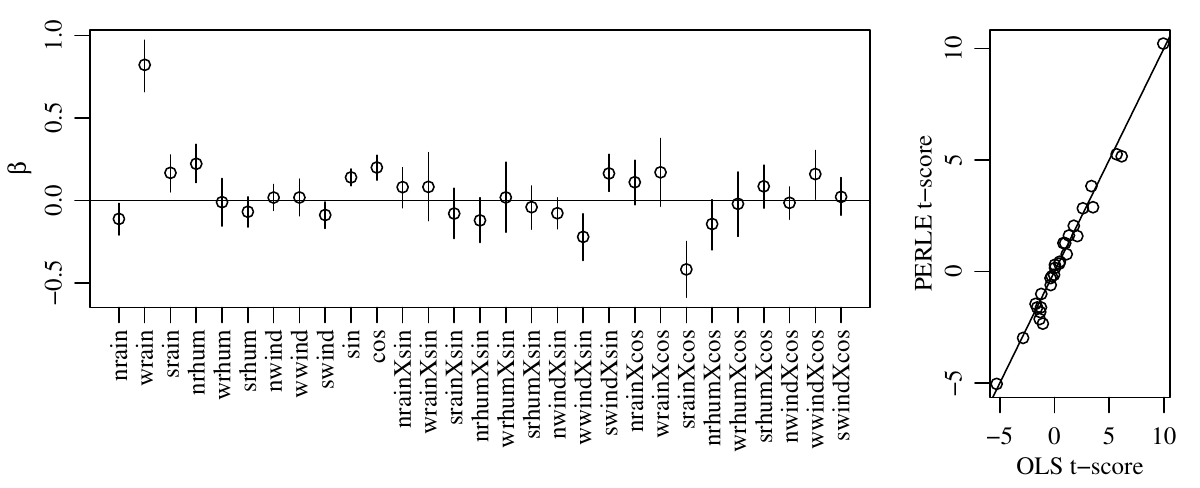} 
\end{knitrout}
\caption{Posterior summary for the Seattle rain data. The left panel displays PERLEs and 95\% PERLE confidence intervals of regression parameters. 
The right panel compares PERLE $t$-scores to those of 
  OLS estimates based on the quarter-power transformation.  } 
\label{fig:sr2} 
\end{figure}

We performed an out-of-sample prediction experiment to evaluate the coverage of 
several prediction interval methods, including the two 
described in Section 4, as well as five other conformal methods. The first two of these  
are simply the standard split-conformal method using the magnitude of the residual deviation as a conformity 
score \citep{papadopoulos_2002}, where the residuals are from a linear model fit to the raw data and to the quarter-power transformed data. 
The next two are based on these same linear model fits, but the 
intervals were calibrated using the locally-weighted
 procedure described in 
\citet{guan_2023}, which uses weighted quantiles of calibration residuals so that 
the prediction interval for a given feature is calibrated primarily by residuals corresponding 
to similar features according to a user-specified kernel function. 
We used a Gaussian kernel function with a bandwidth parameter  chosen by trial and error. 
The last procedure we tried was conformalized quantile regression 
\citep{romano_patterson_candes_2019} as implemented by the {\sf R}-package {\tt probably} \citep{kuhn_vaughan_ruiz_2025}, which fits a quantile random forest model to training data, then adjusts the resulting prediction intervals 
with a calibration dataset.

For each week in the second half of the dataset (starting in July 2021), we used all data preceding that week 
as the training data to fit the model and make predictions for each day in the given week. For 
methods based on conformal calibration, the data from the preceding 365 days 
were used for calibration, 
and the remaining preceding data were used for model fitting. 
We sorted the outcomes into eight bins based on the sample quantiles, resulting 
in about 228 outcomes in each bin,  and computed the 
coverage rate and expected interval width in each bin and for each prediction method using a nominal 80\% 
coverage rate. We use this rate rather than a 95\% rate so that differences between the methods are more clearly distinguished. 

The results of this study are displayed in Figure \ref{fig:sr3}.   
To simplify the figure, the results for the two standard split-conformal intervals are not displayed, as they 
were very similar to their normal-theory counterparts presented in the Introduction, and
were slightly improved upon by their weighted versions. 
Overall, all non-rank-based procedures show overcoverage for most quantile bins but substantial 
undercoverage for the highest quantile bin. 
For the four approaches based on linear models, the biggest improvement in 
conditional coverage results from transforming the outcome, rather than using conformal calibration or weighting. 
In contrast, the Bayes and conformal rank-based approaches have coverage rates that are 
closer to the nominal level across the quantile bins. This is a result of these
procedures automatically adjusting to the mean-variance relationship in these data by having 
interval widths that vary with outcome magnitude, as seen in the plot in the right side of the 
figure. Data transformation,  weighted calibration, and quantile regression methods  provide some degree 
of heterogeneity in interval width, but not to the extent necessary to achieve approximate conditional coverage. 
We note that all methods, including non-calibrated normal-theory intervals, maintained marginal coverage 
rates at or above the nominal level. 

\begin{figure} 
\begin{knitrout}
\definecolor{shadecolor}{rgb}{0.969, 0.969, 0.969}\color{fgcolor}
\includegraphics[width=\maxwidth]{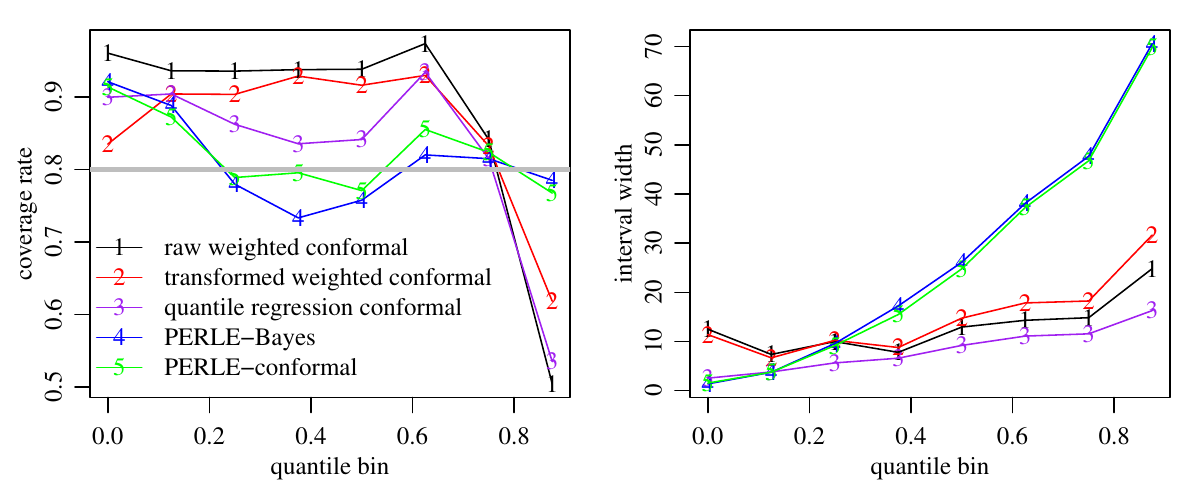} 
\end{knitrout}
\caption{Prediction intervals for the Seattle rainfall data. The left and right panels display conditional coverage rates  and interval widths, respectively, of nominal 80\% intervals of five different methods.}
\label{fig:sr3} 
\end{figure}

\subsection{Income by degree}
The 2024 General Social Survey  \citep{gss2024} provides individual-level
data on social and demographic characteristics of people living in the  United States. 
From these data, we use 
 a sample of size $n=1577$  to 
 model a survey respondent's income as a function of demographic characteristics and educational 
attainment.  
A survey respondent's 
income is recorded as belonging to one of 26  ordered income categories. We model this 
ordinal outcome as a function of the respondent's age (age in years divided by 100), 
binary sex and race variables, highest level of 
educational attainment (high school, associates 
degree, bachelors degree, graduate degree), and degree area (e.g.\ humanities, engineering, social sciences, etc.). 

We ran the Gibbs sampler described in Section 2.2 for 11,000 iterations, dropping the first 
1000 iterations to allow for convergence of the Markov chain, and saving every 10th parameter 
value for posterior approximation. Mixing of the Markov chain was very good, 
with effective sample sizes for all parameters being at least 904.
Posterior mean estimates (PERLEs) and 95\% posterior quantile intervals  
for the regression parameters $\bs\beta$ 
are displayed in the left panel of Figure \ref{fig:ibd1}. 
For comparison, we also obtained a posterior distribution of the regression parameters 
using a full ordinal model that includes a prior distribution for the function $G$, which 
for this ordinal outcome can be parameterized by 25 threshold values. An extremely diffuse 
prior distribution over these thresholds resulted in a posterior distribution for $\bs\beta$ that 
was very similar to that obtained from the posterior using the extended rank likelihood. 
In particular, posterior $t$-scores (posterior means divided by posterior standard deviations) 
obtained from the two methods were nearly identical, as shown in the right panel 
of Figure \ref{fig:ibd1}.

\begin{figure} 
\begin{knitrout}
\definecolor{shadecolor}{rgb}{0.969, 0.969, 0.969}\color{fgcolor}
\includegraphics[width=\maxwidth]{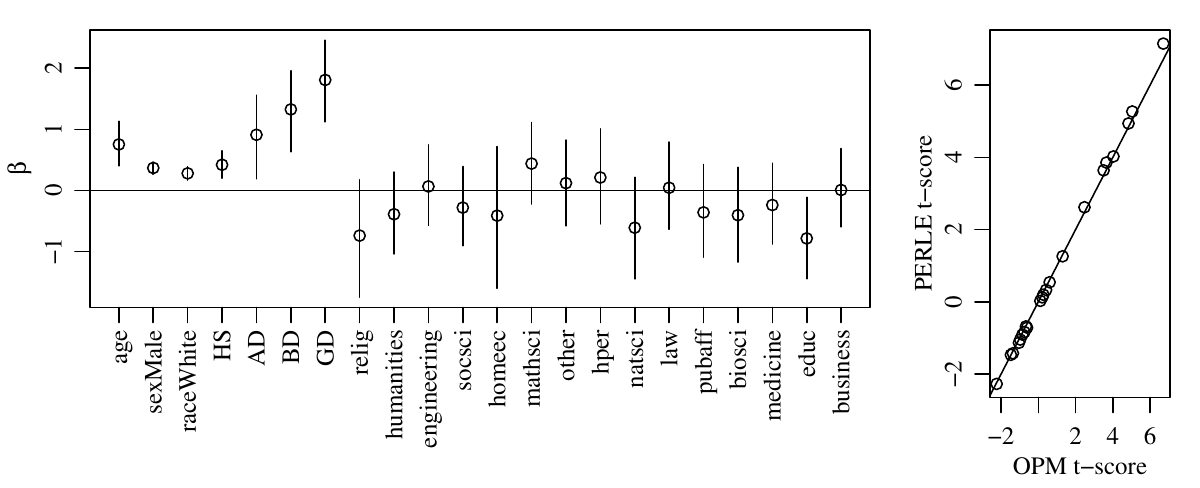} 
\end{knitrout}
\caption{Posterior summary for the income data. The left panel displays the PERLEs and 95\% PERLE confidence intervals for regression parameters. 
The right panel compares PERLE $t$-scores  to those of 
  a full ordinal probit model.  } 
\label{fig:ibd1}
\end{figure}

We performed an out-of-sample prediction experiment to assess the conditional coverage rates 
of the rank-based Bayes and conformal procedures described in Section 4, and 
compared these rates to those of  the full ordinal probit model. Specifically, we constructed 
 $80\%$ prediction intervals for the income category of each subject in the study 
using the data from all other subjects and the three interval procedures (for the conformal procedure, we randomly split data 
from the other subjects into fitting and calibration sets of size 788 each).  
The marginal coverage rates for the Bayes, conformal and full probit model are 85, 82 and  84\% respectively. 
Category-specific coverage rates and interval widths are displayed in Figure \ref{fig:ibd2}. The Bayes and conformal procedures using the extended ranks perform very similarly, and maintain approximate 
conditional coverage across the income categories. Intervals based on the full ordinal probit model 
have approximately correct marginal coverage, but have poor conditional coverage at middle to low income categories, with zero coverage for outcomes in the lowest category.  
We speculate that this is partly due to the highly unbalanced sample sizes in the categories: 
The lowest three categories have a combined sample size of 53, whereas that of the top 
three is 215. For these higher income categories with larger sample sizes, 
the intervals from the full probit model 
 maintain greater than nominal coverage while additionally being narrower than those based on the 
extended ranks. 

\begin{figure} 
\begin{knitrout}
\definecolor{shadecolor}{rgb}{0.969, 0.969, 0.969}\color{fgcolor}
\includegraphics[width=\maxwidth]{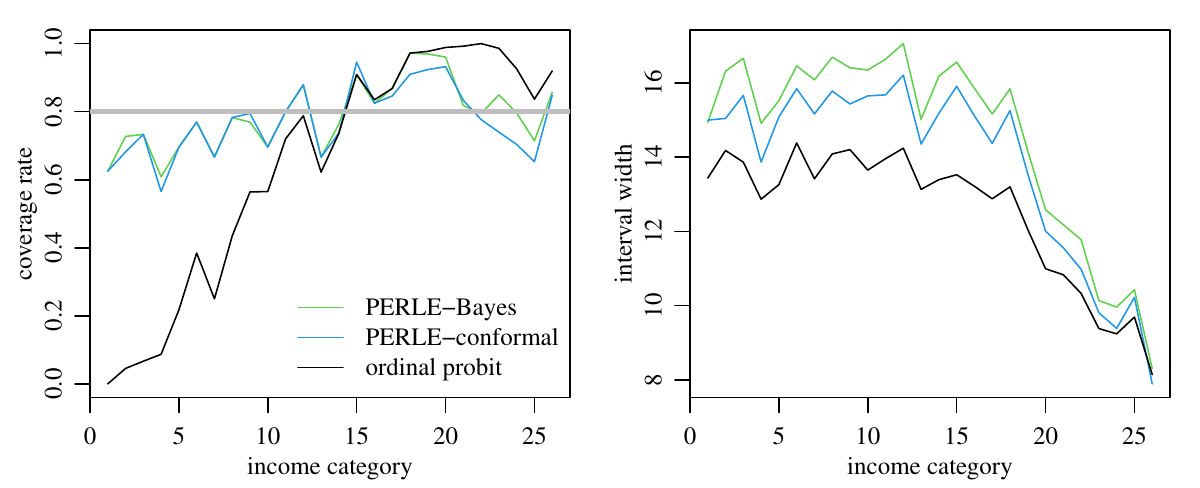} 
\end{knitrout}
\caption{Prediction intervals for the income data. The left and right panels display conditional coverage rates  and interval widths, respectively,  of nominal 80\% intervals of three different methods.}
\label{fig:ibd2}
\end{figure}

\section{Discussion} 
Data transformation and data coarsening (including rounding, binning and truncation) are standard 
preprocessing techniques that can be used to make a preferred model more appropriate and more likely to provide 
accurate inferences. 
However, treating such preprocessing steps as being outside of the modeling process can lead to 
a degree of arbitrariness in the data analysis. 
Data-based or joint estimation of a transformation $G$ along with parameters of interest $\bs\beta$ is possible, but can be difficult if
$G$ is treated nonparametrically, and still requires a modeling decision as to whether the data are to be treated as 
continuous, truncated, discrete or otherwise --- a decision for which there may be more than  one reasonable answer. 
As an alternative to deliberating over a variety of strategies for $G$, 
we advocate for treating 
$G$ as a single nuisance parameter via rank-based methods and the extended rank likelihood. 
We have discussed how the rank likelihood incurs no asymptotic loss in efficiency, 
and have shown how rank-based methods may be used to construct
approximately conditionally-calibrated prediction intervals, all 
without having to specify a parametric or nonparametric transformation family or prior distribution for $G$, or make a determination of whether 
the outcome variable is continuous or discrete. 

Our prediction experiments indicate that, while different models or data transformations
may all result in the same marginal prediction coverage rate, 
the conditional performance of procedures may vary considerably. 
Our theoretical results from Section 4 indicate that, under the MTLM,  rank-based prediction intervals can achieve 
asymptotically calibrated conditional coverage rates, regardless of what the appropriate transformation is. 
The analogy to conformal prediction is that  marginally calibrated 
intervals are provided by all choices of conformity score, whereas conditional calibration depends 
on the choice of score. Different scores may correspond to different models: A score based on 
residual magnitude corresponds to a homoscedastic model, locally-weighted calibration 
\citep{guan_2023} presumes a smooth mean-variance relationship, conformalized quantile
regression \citep{romano_patterson_candes_2019} presumes an adequate quantile model, 
and the rank-based conformity score we propose 
presumes the MTLM. 
In this sense, choosing a conformity score is not a neutral step preceding the
modeling, but is itself an act of modeling, with the assumptions entering
implicitly.

\appendix

\section*{Proofs} 
The proof of Theorem \ref{thm:rmle-sqrt-n} is much longer than the others, 
and so is presented last.

\begin{proof}[Proof of Lemma \ref{lem:xrankG}]
Let $r(\bl y)_i=\{ l,\ldots, u\}$, which may be a singleton if $l=u$. Then there are exactly $l-1$ elements of 
$\bl y$ strictly less  than $y_i$ and $n-u$ strictly greater, and so by the monotonicity of $G$,
 $l \leq r(\bl z)_i  \leq u$. 
\end{proof} 

\begin{proof}[Proof of Lemma \ref{lem:ccineq}]
We want to bound the variance of $\bs Z$ under the density 
$p(\bl z) = c(S) \exp\{ - ||\bl z-\bs\mu ||^2/2 \} 1(\bl z\in S)$. 
This can be done by approximating $p(\bl z)$ by a $C^2$ strictly log concave function, applying the 
Brascamp--Lieb variance inequality \citep{brascamp_lieb_1976}, and then taking some limits. 
Letting $\bar S$ be the closure of $S$, 
define $f(\bl z) = ||\bl z - \bs\mu ||^2/2$ for $\bl z\in \bar S$ and $f(\bl z ) = \infty$ for 
$\bl z\not \in \bar S$, so that $p(\bl z) \propto e^{-f(\bl z)}$. 
First approximate $f(\bl z)$ with   
$f_{\epsilon} =  ||\bl z-\bs\mu ||^2/2  + 
  m_\epsilon(\bl z)$, where 
$m_\epsilon(\bl z) = (2\epsilon)^{-1}  \min_{\bl s \in \bar S} || \bl z -\bl s||^2$ 
is the Moreau envelope of the function that is zero on $\bar S$ and $\infty$ off of $\bar S$, 
which is convex because $S$ and $\bar S$ are convex. 
As a result, $f_\epsilon\in C^{1,1}$ is strictly convex for each 
$\epsilon>0$ and converges pointwise to $f$ as $\epsilon \rightarrow 0$. 
However, the Brascamp--Lieb inequality applies to densities proportional to $e^{-g(\bl z)}$ where 
$g\in C^2$, and so we must further smooth the approximation of $f$.  
Let $f_{\epsilon,\delta}= ||\bl z-\bs \mu||^2/2 + m_{\epsilon,\delta}(\bl z)$ where
\[ m_{\epsilon,\delta}(\bl z ) =  ( k_{\delta} \ast m_\epsilon )(\bl z)     = 
  \int m_{\epsilon}(\bl w)  k( (\bl z - \bl w )/\delta )/\delta^n  \,  d\bl w \]
for a $C^2$ kernel density $k$. 
Then $f_{\epsilon, \delta} \in C^2$ is strictly convex and converges to $f_{\epsilon}$ pointwise as 
$\delta\rightarrow 0$. 

By the Brascamp--Lieb inequality, for a unit vector $\bl u$ the variance of $\bl u^\top \bs Z$
 under the distribution with density 
$p_{\epsilon,\delta}  \propto e^{- f_{\epsilon,\delta}(\bl z)}$
is bounded by one: 
\begin{align*} 
 \text{Var}_{\epsilon,\delta}[ \bl u^\top \bs Z] \leq  \int \bl u^\top H^{-1}_{\epsilon,\delta} \bl u \,  
   p_{\epsilon,\delta}(\bl z) \, d\bl z   \leq 1. 
\end{align*} 
The first inequality follows from the Brascamp--Lieb inequality, and the second because 
 $H_{\epsilon,\delta}$, the Hessian of $f_{\epsilon,\delta}$,  
satisfies $H_{\epsilon,\delta}^{-1} \preceq \bl I_n$ as $m_{\epsilon,\delta}$ is convex. 
Now $\text{Var}_{\epsilon,\delta}[ \bl u^\top \bs Z]$ depends on the first and second 
moments of $\bs Z$ under this density, which can be computed from the unnormalized expectations 
\[
  \int (\bl u^\top \bl z)^l  e^{ -f_{\epsilon,\delta} (\bl z) } \, d\bl z  
\]
for  $l\in \{0,1,2\}$. 
These integrals are dominated by the corresponding integrals with respect to $e^{-||\bl z- \bs\mu||^2/2} $, and 
the integrands converge pointwise to $(\bl u^\top \bl z)^l e^{-f_{\epsilon}(\bl z) }$, and 
so by the dominated convergence theorem (DCT), they converge to the integrals of 
their pointwise limits. Furthermore, the normalizing constant  converges to $\int e^{-f_\epsilon(\bl z)}\, d\bl z  \geq \int_{\bar S} e^{-||\bl z-\bs\mu||^2/2} \, d\bl z >0$ by the DCT, and so is bounded away from zero for 
sufficiently small $\delta$. As a result, 
$\text{Var}_{\epsilon,\delta}[ \bl u^\top \bs Z]   \rightarrow \text{Var}_{\epsilon}[ \bl u^\top \bs Z]$ 
as $\delta\rightarrow 0$, where the latter variance is with respect to the distribution with density 
$p_{\epsilon}(\bl z) \propto e^{-f_\epsilon(\bl z)}$. Thus we have $\text{Var}_{\epsilon}[ \bl u^\top \bs Z] \leq 1$ 
for all unit vectors $\bl u$.
Finally, we apply exactly the same DCT argument to this variance under $p_\epsilon$ as $\epsilon\rightarrow 0$ to
obtain $\Var{ \bl u^\top \bs Z | \bs Z \in \bar S} = \Var{ \bl u^\top \bs Z | \bs Z \in S}\leq 1$, 
since the boundary of $S$ has measure zero. This gives the result. We note that the same proof can be applied to the case that $\bs Z\sim N_n( \bs\mu, \Sigma)$ for arbitrary nonsingular covariance matrices $\Sigma$,  giving 
$\Var{ \bs Z|\bs Z\in S} \preceq \Sigma$. 
\end{proof}

\begin{proof}[Proof of Theorem \ref{thm:effZero}]
Let $\bs R$ and $\bs Z_{()}$  be the vector of ranks and order statistics of $\bs Z$, 
respectively.  
In the case that $\bs\beta=0$, the elements of $\bs Z$ are i.i.d.\ and hence $\bs R$ and $\bs Z_{()}$ 
are independent. Additionally, since $G$ is strictly increasing, the conditional distributions 
$\bs Z | \bs Z \in S$ and $\bs Z|\bs R$ are equivalent. Hence 
\[
 \Var{\bs Z | \bs Z \in S } = \Var{ \bs Z | \bs R } = \bl P \Var{\bs Z_{()}} \bl P^\top  \equiv \bl V, 
\]
where $\bl P$ is the matrix that permutes the indices of the order statistics to their observed orders, which 
is a function of $\bs R$. Note that $\Var{\bs Z_{()}}$ is deterministic while $\bl P$ depends on the ranks 
of $\bs Z$, which are equal to the ranks of $\bs Y$. 

Now let $\bar Z= \bl 1^\top \bs Z/n$ and $\bl e= \bs Z - \bar Z \bl 1$. Note that $\bar Z$ and $\bl e$ are independent, and that $\bs R$ is a function of $\bl e$, and so $\bl V =  \bl 1\bl 1^\top/n + \Var{\bl e| \bs R}$. Defining 
$\bl W = \bl V - \bl 1 \bl 1^\top /n = \Var{\bl e| \bs R}$, we have that $\bl W$ is a positive semidefinite 
matrix with 
 $\bl W \bl 1 = \Cov{ \bl e, \bl e^\top \bl 1 | \bl R } = \bl 0$ and 
 $\tr(\bl W )= \tr(\bl V ) -1$ being nonnegative and deterministic. 
The moments of the elements of $\bl V$ and $\bl W$ can be computed as follows: Since 
$V_{i,i} = \Var{ Z_{(R_i)}}$, $V_{i,j} = \Cov{ Z_{(R_i)}, Z_{(R_j)} }$  and 
$R_i$ is uniform on $\{1,\ldots, n\}$, we have 
\begin{align*}  \Exp{ V_{i,i} } & = \sum_{k=1}^n \Var{ Z_{(k)}}/n = \tr(\bl V) /n  \\
\Exp{ V_{i,j}} & = ( \bl 1^\top \Var{\bs Z_{()} }\bl 1  - \tr(\Var{\bs Z_{()} } ) )/(n(n-1)) = (n-\tr(\bl V))/(n(n-1)),
\end{align*} 
giving $\Exp{W_{i,i}} = \tr(\bl W)/n$  and $\Exp{ W_{i,j} }= -\tr(\bl W)/(n(n-1))$. 
Additionally, we have 
\begin{align*} 
\tr(\bl V) & = \sum_{i=1}^n \left(  \Exp{ Z_{(i)}^2}  - \Exp{ Z_{(i)} }^2 \right )  \\
     & = \sum_{i=1}^n \Exp{ Z_i^2} - \sum_{i=1}^n \Exp{Z_{(i)}}^2 = n - \sum_{i=1}^n  \Exp{ Z_{(i)} }^2. 
\end{align*} 
We now show $\tr(\bl V) /n \rightarrow 0$ and consequently $\tr(\bl W)/n \rightarrow 0$. 
Let $Q_n(t)$ be the empirical quantile function of $\bs Z$, so $Q_n(t) = Z_{(i)}$ for $t\in ((i-1)/n,i/n]$ and
 $\int_0^1 \Exp{Q_n(t)}^2 \, dt = \sum_{i=1}^n \Exp{ Z_{(i) }}^2/n$.  
We have 
\[ \int_0^1 ( \Exp{Q_n(t)} - \Phi^{-1}(t) )^2 \, dt \leq 
    \Exp{ \int_0^1 (Q_n(t) - \Phi^{-1}(t))^2   \,  dt }  \rightarrow 0 
\]  
as $n\rightarrow \infty$ from Jensen's inequality and Theorem 5.1 of \citet{bobkov_ledoux_2019}, so 
$\Exp{ Q_n } \rightarrow \Phi^{-1}$ in $L_2(0,1)$ as $n\rightarrow \infty$. This implies
\[
  \left | \int_0^1 \Exp{Q_n(t) }^2  \, dt -  \int_0^1 ( \Phi^{-1}(t))^2 \, dt \right | = 
  \left |\sum_{i=1}^n \Exp{ Z_{(i)}}^2/n - 1 \right | \rightarrow 0 
\]
as $n\rightarrow \infty$, which gives the result. 
Finally, since $\bl X$ and $\bl W$ are independent in the case that $\bs\beta=0$, we have 
\[ \Exp{ \tr (\bl X^\top \bl W \bl X) /n } =\left( \sum_{j=1}^p \Var{x_{1,j}} \right )\times \tr(\bl W)/n \rightarrow 0\] 
as $n\rightarrow \infty$. Because $\bl X^\top \bl W \bl X$ is nonnegative definite, 
we also have 
$\tr (\bl X^\top \bl W \bl X) /n$ converging in probability to zero by Markov's inequality. Every entry of $\bl X^\top \bl W \bl X/n$ 
is bounded by the trace, and so 
 $\bl X^\top ( \Var{\bs Z | \bs Z\in S } - \bl 1\bl 1^\top /n ) \bl X/n$ converges in probability to zero 
elementwise. 

\end{proof} 

\newcommand{\zgq}{\ensuremath{ Z_{(\lfloor \gamma n \rfloor ) } }} 

\begin{proof}[Proof of Lemma \ref{lem:lipschitz}]
Because $Z$ and $\bs Z$ are independent, the function $H(\gamma,\bl b)$ can be expressed 
\begin{align*}
H(\gamma,\bl b) & = \Exp{ \Phi( Z_{(\lfloor \gamma n  \rfloor )} - \bl x^\top \bl b )  |\bl b } \\
   & = \Exp{ \Phi( (\bl X\bl b + \bs\epsilon )_{(\lfloor \gamma n  \rfloor )} - \bl x^\top \bl b )   }     \equiv  \Exp{ \Phi( h(\bl b, \bs\epsilon)) }, 
\end{align*}
where $\Phi$ is the standard normal CDF and the latter 
expectation is over $\bs\epsilon \sim N_n(\bl 0, \bl I_n)$.
Therefore 
\begin{align*}
|H(\gamma,\bs \beta)  - H(\gamma , \tilde {\bs\beta} ) |   &= 
 | \Exp{ \Phi( h(\bs \beta,\bs\epsilon)) - \Phi(h(\tilde{\bs \beta} , \bs \epsilon ) ) }| \\
 &  \leq  \Exp{ | \Phi( h(\bs \beta,\bs\epsilon)) - \Phi(h(\tilde{\bs \beta} , \bs \epsilon ) ) | } \\
\end{align*} 
Since $\Phi$ is Lipschitz with constant $(2\pi)^{-1/2}$, we have 
$| \Phi( h(\bs \beta,\bs\epsilon)) - \Phi(h(\tilde{\bs \beta} , \bs \epsilon ) ) | \leq 
(2\pi)^{-1/2} | h(\bs \beta,\bs\epsilon)  - h(\tilde{ \bs \beta},\bs\epsilon) |$. 
Now recall that for vectors $\bl u$ and $\bl v$, 
$\max_k |u_{(k)} - v_{(k)}| \leq  ||\bl u - \bl v ||_\infty$, from which it follows that 
\begin{align*} 
| h(\bs \beta,\bs\epsilon)  - h(\tilde{ \bs \beta},\bs\epsilon) | 
& \leq || (\bl X\bs\beta + \bs\epsilon ) - (\bl X\tilde{\bs\beta} + \bs\epsilon ) ||_{\infty} + 
   |\bl x^\top(\bs\beta-\tilde{\bs\beta} ) | \\
 & = || \bl X(\bs\beta -\tilde{\bs\beta}) ||_{\infty} + 
   |\bl x^\top(\bs\beta-\tilde{\bs\beta} ) |  
\end{align*}
which does not depend on $\bs\epsilon$. This gives the result. 
\end{proof}

\begin{proof}[Proof of Corollary \ref{cor:rootnH}] 
From Lemma \ref{lem:lipschitz}, 
$|| \bl X(\bs\beta -\hat{\bs\beta}) ||_{\infty}  \leq \max \{ ||\bl x_1||,\ldots, 
 ||\bl x_n|| \}  \times ||\bs\beta - \hat{\bs \beta} || = O_p(n^{-1/2})$ under  
the  assumption on $\bl x_1,\ldots, \bl x_n$. 
The second term is 
$ |\bl x^\top(\bs\beta-\hat{\bs\beta} ) |   \leq ||\bl x|| ||\bs\beta - \hat{\bs\beta} || = O_p(n^{-1/2})$.  Since the Lipschitz bound from Lemma \ref{lem:lipschitz} does not depend on $\gamma$, the 
supremum of  $|H(\gamma,\bs\beta) - H(\gamma,\hat{\bs\beta}) |$ over $\gamma$ satisfies the same rate. 
\end{proof}

\begin{proof}[Proof of Theorem \ref{thm:hconv}] 
Write 
\begin{align*} 
H(\hat\gamma_u, \bs\beta )- H(\hat\gamma_l, \bs\beta )  = 
H(\hat\gamma_u, \hat{\bs\beta} )- H(\hat\gamma_l, \hat{\bs\beta} )  &  +
( H(\hat\gamma_u, \bs\beta )- H(\hat\gamma_u, \hat {\bs\beta }) )      \\
  & - ( H(\hat\gamma_l, \bs\beta )- H(\hat\gamma_l, \hat {\bs\beta }) ).   
\end{align*}
Under the assumptions of the theorem, the difference of the first two terms on the right side 
converges in probability to $1-\alpha$, while the two remaining terms in parentheses converge  
in probability to zero because the convergence of $H(\gamma,\hat{\bs\beta})$ to  
$H(\gamma,{\bs\beta})$ is uniform in $\gamma$. 
\end{proof}

\begin{proof}[Proof of Theorem \ref{thm:cconv}]
The coverage probability can be written 
\begin{align*}
\Pr( Z_{(\lfloor \hat \gamma_l n \rfloor )} < Z < 
 Z_{(\lceil \hat \gamma_u n \rceil )} | \bs\beta , \bl x)     
 & = \Exp{ \Phi(Z_{(\lceil \hat\gamma_u n\rceil)} - \bl x^\top \bs\beta ) - \Phi(Z_{(\lfloor \hat\gamma_l n\rfloor )}   - \bl x^\top\bs\beta)
  }. 
\end{align*}  
Consider the first term on the right side. Letting $r(\bs Z)$ denote the ranks 
of $Z_1,\ldots, Z_n$, 
\begin{align*} 
\Exp{ \Phi(Z_{(\lceil \hat\gamma_u n\rceil)} - \bl x^\top \bs\beta )  | \bs\beta} & = 
\Exp{ \Exp{  \Phi(Z_{(\lceil \hat\gamma_u n\rceil)} - \bl x^\top \bs\beta ) |\bs\beta, r(\bs Z)  } | \bs\beta }  \\ 
& = \Exp{ H( \hat \gamma_u,\bs\beta ) | \bs\beta } 
\end{align*} 
because $\hat\gamma_{u}$ is rank-measurable, $Z_{(\lceil \hat\gamma_u n\rceil)}$ is 
an order statistic, and 
 for an i.i.d.\ sample of continuous random variables   the ranks and order statistics  
are independent. 
The same argument applies to the second term on the right side of the equation. 
Now $H( \hat \gamma_u,\bs\beta )  -
  H( \hat \gamma_l,\bs\beta )$ is bounded  and 
by Theorem \ref{thm:hconv} converges in probability to $1-\alpha$, so the
expected difference converges to $1-\alpha$ as well. 

\end{proof}

\begin{proof}[Proof of Theorem \ref{thm:rmle-sqrt-n}]
As described before the statement of the proof, the ERL 
may be written as 
\begin{align*} 
L(\bs\beta: S(\bl y)) & = 
 \int 
 \left\{
\prod_{i=1}^n
  \Phi(\alpha-\bl{x}_i^{\top}\bs{\beta})^{1-y_i} (1-\Phi(\alpha-\bl{x}_i^\top\bs{\beta}))^{y_i } \right \}\left  (\sum_{k\in N_0} \frac{  \phi(\alpha-\bl{x}_k^\top\bs{\beta})}{\Phi(\alpha-\bl x_k^\top \bs\beta )} \right )   \, d\alpha
\\ 
& \equiv \int L( \bs\theta : \bl y) w(\bs\theta ) \, d\alpha
\end{align*} 
where $\bs\theta= (\alpha,\bs\beta)$, and $L(\bs\theta:\bl y)$ is the usual probit regression likelihood function 
for $\bs\theta$ given the data vector $\bl y$. 
In what follows, candidate parameter values are expressed as $\bs\theta= (\alpha,\bs\beta)$ and 
the true values that generate the data are expressed as $\bs\theta_0 = (\alpha_0,\bs\beta_0)$. 

Given a prior distribution over $\bs \beta$ with density $\pi_\beta(\bs\beta)$, define the PERLE
 $\hat {\bs\beta}$ as the expectation of $\bs\beta$ under the probability distribution with 
density proportional to $\pi_\beta (\bs\beta ) L(\bs\beta:S(\bl y))$. 
It is easily shown that $\hat{\bs\beta}$ is the $\bs\beta$-component of the posterior mean estimator 
$\hat{\bs\theta}$ of $\bs\theta$, given by 
\[ 
\hat{\bs\theta } = \frac{ \int \bs\theta L(\bs\theta:\bl y ) \pi(\bs\theta)\, d\bs\theta }{ 
                          \int  L(\bs\theta:\bl y ) \pi(\bs\theta)\, d\bs\theta }, 
\]
where $\pi (\bs\theta) \propto \pi_\beta(\bs\beta) w(\bs\theta)$ is interpreted as pseudo-prior distribution, keeping in 
mind that $\pi$ depends on $\bl y$ through $w(\bs\theta)$. 
Letting $\tilde{\bs \theta}$ be the maximizer of $L(\bs\theta:\bl y)$, 
 we will show that $\sqrt{n}( \hat{\bs\theta} - \tilde {\bs\theta}) = o_p(1)$. 
Furthermore, under the assumptions  of the theorem, $\tilde{\bs\theta}$ satisfies $\sqrt{n}(\tilde{\bs\theta} - \bs\theta_0) \stackrel{d}{\rightarrow} 
N_{p+1}( \bl 0 , I(\bs\theta_0)^{-1})$, where 
$I(\bs\theta_0)$ is the Fisher information of the probit regression
model with expectation taken jointly over the distribution of $(\bl x,Y)$ 
\citep{fahrmeir_kaufmann_1986}. Together with the first result, this implies that $\sqrt{n}(\hat{\bs\beta} - \bs\beta_0) \stackrel{d}{\rightarrow} N_p(0,(I(\bs\theta_0)^{-1})_{\beta\beta})$. 

Let $l(\bs\theta) = \log L(\bs\theta:\bl y)$ be the log likelihood, and define 
$\Delta(\bl s) = l(\tilde {\bs\theta} + \bl s/\sqrt{n} ) - l(\tilde{\bs\theta})$. Then  
\begin{equation}
\label{eqn:thetaHat}
\sqrt{n} (\hat{\bs\theta} - \tilde{\bs\theta})  = 
 \frac{  \int \bl s e^{\Delta(\bl s)} \rho(\bl s) \, d\bl s } 
      {   \int      e^{\Delta(\bl s)} \rho(\bl s) \, d\bl s },  
\end{equation} 
where $\rho(\bl s) = \pi(\tilde{\bs\theta} + \bl s/\sqrt{n})/\pi(\tilde{\bs\theta})$.
Note that $\rho$ is positive and $\rho(\bl 0) = 1$.  We will 
show that $e^{\Delta(\bl s)} \rho(\bl s)$ converges to a centered Gaussian kernel so that 
the numerator of (\ref{eqn:thetaHat}) converges  to the first moment of a symmetric function, which is zero, 
while the denominator remains bounded away from zero. 
To do this we make use of two lemmas regarding the asymptotic behavior of $e^\Delta$ and $\rho$.
Starting with the first term, via a Taylor series expansion we show that
 $\Delta(\bl s)$ is close to $-\bl s^\top I(\bs\theta_0)\bl s/2$. 
Specifically, we have the following results, which are proven after the theorem:

\begin{lemma}
\label{lem:delta}  
Let  $\lambda_0$ be
the smallest eigenvalue of $I(\bs \theta_0)$.  
For each fixed $T>0$, 
\begin{itemize} 
\item $\sup_{\bl s : \|\bl s\|\leq T }|\Delta(\bl s) +  \bl s^\top I(\bs\theta_0)\bl s/2| \stackrel{p}{\rightarrow}  0$;
\item $\ \Pr\big( \Delta(\bl s) \leq - \lambda_0  T\|\bl s\|/8 \text{ for all } \|\bl s\| \geq T \big) \to 1$. 
\end{itemize}
\end{lemma}

Regarding $\rho$, it is asymptotically flat where the likelihood concentrates, and grows only polynomially elsewhere:
\begin{lemma}
\label{lem:rho}
For each fixed $T>0$,
\begin{itemize}
\item $\ \sup_{\bl s : \| \bl s \|\leq T }|\log\rho(\bl s)| = O_p(T/\sqrt n)$, and in particular
$\sup_{\bl s : \|\bl s\|\leq T }|\rho(\bl s) - 1| \stackrel{p}{\to} 0$.
\item There exist $C_n = O_p(1)$, $\varepsilon_n = O_p(n^{-1/2})$, and events $G_n$ with $\Pr(G_n) \to 1$, such that on
$G_n$, $\ \rho(\bl s) \leq C_n(1+\|\bl s\|)e^{\varepsilon_n\|\bl s\|}$ for all $\bl s \in \mathbb R^{p+1}$.
\end{itemize}
\end{lemma}

We use the first items of each lemma  to control the numerator and denominator of 
(\ref{eqn:thetaHat})  on the set $\mathcal S_T = \{ \bl s : \|\bl s\|\leq T \}$. 
Note that both $\Delta(\bl s)$ and $-\bl s^\top I(\bs\theta_0) \bl s/2$ are less than or equal to zero, 
and so 
\begin{align}
\label{eqn:prodbd}
  |e^{\Delta(\bl s)}\rho(\bl s) - e^{-\frac12 \bl s^\top I(\bs\theta_0)\bl s}|
  \ & \leq  e^{\Delta(\bl s)}\big|\rho(\bl s) - 1\big| + |e^{\Delta(\bl s)} - e^{-\frac12\bl s^\top I(\bs\theta_0)\bl s}|  \nonumber \\
  \ & \leq \big|\rho(\bl s) - 1\big| + |\Delta(\bl s) + \tfrac12\bl s^\top I(\bs\theta_0)\bl s|
\end{align}
for every $\bl s \in \mathbb R^{p+1}$.  
Since $\mathcal S_T$ is bounded,  (\ref{eqn:prodbd}) and the first items of the two lemmas give uniform convergence of $e^{\Delta(\bl s)} \rho (\bl s) $ to $e^{-\bl s^\top I(\bs\theta_0)\bl s/2}$ on $\mathcal S_T$,  and so we have
\begin{align*}
  \int_{\mathcal S_T}\bl s\, e^{\Delta(\bl s)}\rho(\bl s)\, d\bl s & \stackrel{p}{\rightarrow}  \int_{\mathcal S_T}\bl s\, e^{-\bl s^\top I(\bs\theta_0)\bl s/2}\, d\bl s \ =\ \bl 0 \text{ and }  \\
  \int_{\mathcal S_T} e^{\Delta(\bl s)}\rho(\bl s)\, d\bl s \  &  \stackrel{p}{\rightarrow}  \int_{\mathcal S_T} e^{-\bl s^\top I(\bs\theta_0) \bl s/2}\, d\bl s \equiv c_T, 
\end{align*}
the first limit vanishing because $\mathcal S_T$ is symmetric about the origin and
$\bl se^{-\bl s^\top I(\bs\theta_0) \bl s/2}$ is an odd function of $\bl s$. 
Note that the constant $c_T$ is strictly increasing in $T$, the integrand being positive and the domains nested, so 
for $T>1$ we have
$c_T >  c_1 > 0$. Furthermore, note that the second limit provides the desired bound 
on the denominator of (\ref{eqn:thetaHat}), as the integral over $\mathbb R^{p+1}$ will be larger than that
over just $\mathcal S_T$. Thus for $T>1$ the denominator converges in probability to something greater than or equal to $c_1$, which is greater than zero. 

What remains is to control the numerator of (\ref{eqn:thetaHat}) on the complement of $\mathcal S_T$. 
Let $E_n(T)$ be the event 
that $\Delta(\bl s) \leq - \lambda_0 T\|\bl s\| /8 \text{ for all } \|\bl s\| > T$, 
let $F_n(T)$ be the event that $\varepsilon_n < \lambda_0 T /16$, and let $G_n$ be as described in 
Lemma \ref{lem:rho}. The probabilities of these events all converge to 1, and so the probability of their 
intersection 
$H_n(T) = E_n(T) \cap F_n(T) \cap G_n$  also converges to 1. 
Therefore, on the event $H_n(T)$, we have 
\begin{align*} 
  \| \bl s \| e^{\Delta(\bl s)}\rho(\bl s)
  & \leq  \|\bl s \| e^{- \lambda_0 T\|\bl s \|/8}\cdot C_n(1+\|\bl s\|)e^{\varepsilon_n\|\bl s\|}  \\
  &  \leq   C_n( \|\bl s \| +\|\bl s\|^2)e^{-\lambda_0 T\|\bl s\|/16} \\
  &  \leq   C_n(1+\|\bl s\|)^2  e^{-\lambda_0 T\|\bl s\|/16}. 
\end{align*} 
Integrating, we have
\[
  \int_{\|\bl s\|> T}\|\bl s\|\,e^{\Delta(\bl s)}\rho(\bl s)\,d\bl s
  \ \leq  C_n 
   \int_{\|\bl s\|> T}(1+\|\bl s\|)^2 e^{-\lambda_0 T\|\bl s\|/16 }\,d\bl s  
\equiv C_n \eta(T), 
\]
where $\eta(T)$ is deterministic, finite for each $T$, and satisfies $\eta(T) \to 0$ as $T \to \infty$.

Now pick any $\epsilon>0$ and $\delta>0$. 
From Lemma \ref{lem:rho}, there exists $M>0$ and $n_1$ such that for $n>n_1$,
$\Pr( C_n > M )\leq \delta/4$, and a 
 $T>1$  such that $M \eta(T)<\epsilon/2$, so that 
 $U_n(T) \equiv \int_{\|\bl s\|> T} \|\bl s\| e^{\Delta(\bl s)} \rho(\bl s) \, d\bl s\leq M \eta(T) <\epsilon/2$
on the event $H_n(T) \cap \{ C_n \leq M\}$. 
We then have
\begin{align*}  
\Pr( U_n(T) > \epsilon/2)&\leq\Pr( \{ U_n(T) > \epsilon/2 \} \cap H_n(T) \cap \{ C_n \leq M\} )+ \\
     &  \qquad \Pr(H_n(T)^c) + \Pr( C_n>M) \\
& \leq   \Pr(\emptyset ) + \delta/4 + \delta/4    \\
 & = 0 + \delta/4 + \delta/4   = \delta/2  
\end{align*} 
for $n>n_1 \vee n_2$, where $n_2$ is chosen so that $\Pr( H_n(T) ) >1-\delta/4$ for $n >n_2$. 
Let $W_n$ be the norm of the  numerator in (\ref{eqn:thetaHat}) and 
$V_n(T)  = 
| \int_{\mathcal S_T}  \bl s e^{\Delta(\bl s) } \rho(\bl s) \, d \bl s|$, which was previously shown to be $o_p(1)$. 
By the triangle inequality we have $W_n \leq  U_n(T) + V_n(T)$, giving
\begin{align*} 
\Pr( W_n>\epsilon ) &\leq  \Pr(  U_n(T) + V_n(T) > \epsilon ) \\
 & \leq \Pr( U_n(T)>\epsilon/2 ) + \Pr(V_n(T)>\epsilon/2 )  < \delta 
\end{align*}  
for $n> n_1\vee n_2\vee n_3$, where $n_3$ is such that $\Pr( V_n(T)> \epsilon/2)<\delta/2$ for all $n>n_3$. 
Hence the numerator of (\ref{eqn:thetaHat}) is $o_p(1)$, 
 the denominator of (\ref{eqn:thetaHat})  is bounded from below by $c_1>0$ in probability, 
and so  $\sqrt{n} (\hat{\bs\theta} - \tilde{\bs\theta} ) = o_p(1)$
and in particular $\sqrt{n} (\hat{\bs\beta} - \tilde{\bs\beta} ) = o_p(1)$. By Slutsky's theorem and 
the efficiency of $\tilde{\bs\beta}$, the result follows. 
\end{proof} 

We now prove Lemmas \ref{lem:delta} and \ref{lem:rho}. To prove the former, we first obtain a bound on the 
derivatives of the log likelihood terms. In what follows, 
$\phi$ and $\Phi$ denote the standard normal density and distribution functions, and 
  $\psi  = \phi/\Phi  = (\log\Phi)'$
denotes the logarithmic derivative of $\Phi$. 
Two properties of $\psi$ are used repeatedly: The first is that 
\begin{equation}
\label{eqn:psiprime}
  \psi'(u) \ =\ (\log\Phi)''(u) \ \in\ (-1,0) \qquad \text{for all } u, 
\end{equation}
so that $\psi$ is decreasing and $|\psi'| \leq 1$. Writing $\psi' = -\psi\,(u+\psi)$, the two halves of 
\eqref{eqn:psiprime} are equivalent to the classical hazard-rate inequalities $\psi(u) > -u$ and 
$\psi(u)\{\psi(u)+u\} < 1$ for the standard normal distribution \citep{sampford_1953}. 
The second property follows from the first: $\psi$ is decreasing with 
$\psi(0) = \phi(0)/\Phi(0) = \sqrt{2/\pi}$, so $\psi(u) \leq \sqrt{2/\pi}$ for $u \geq 0$, while for $u<0$ 
\[
  \psi(u) \ =\ \psi(0) + \int_u^0 \{-\psi'(t)\}\, dt \ \leq \sqrt{2/\pi} + |u| , 
\]
and therefore 
\begin{equation}
\label{eqn:psibound}
  \psi(u) \ \leq \sqrt{2/\pi} + \max(0,-u) \ \leq 1 + |u| \qquad \text{for all } u. 
\end{equation}

\begin{lemma}\label{lem:deriv-bd}
For $m \in\{1,\ldots, 4\}$ there is a constant $C_m$ such that, for \emph{all} $\alpha \in \mathbb{R}$,
$\bs\beta \in \mathbb{R}^p$, and $\bl x \in \mathbb{R}^p$, and for every multi-index with $r$ derivatives in
$\alpha$ and one derivative in each of $\beta_{j_1},\dots,\beta_{j_s}$ with $r + s = m$,
\[
  \left|\frac{\partial^m}{\partial \alpha^{r}\,\partial\beta_{j_1}\cdots\partial\beta_{j_s}}
        \log\Phi(\alpha - \bl x^\top \bs\beta)\right|
  \ \leq C_m\,\|\bl x\|_s^{\,s}\,\bigl(1 + |\alpha - \bl x^\top \bs\beta|^{m}\bigr),
\]
and the same bound holds with $\log\Phi$ replaced by $\log\bigl(1 - \Phi(\alpha - \bl x^\top \bs\beta)\bigr)$.
\end{lemma}

\begin{proof}
For a given multi-index each $\alpha$-derivative contributes a
factor of $1$ and each $\beta_{j}$-derivative contributes a factor of $-x_{j}$, so
\[
  \frac{\partial^m}{\partial \alpha^{r}\,\partial\beta_{j_1}\cdots\partial\beta_{j_s}}
    \log\Phi(\alpha - \bl x^\top \bs\beta)
  = (-1)^{s}\Bigl(\textstyle\prod_{i=1}^{s} x_{j_i}\Bigr)\, g^{(m)}(u),
\]  
where $g = \log\Phi$ and $u = \alpha - \bl x^\top\bs\beta$.
Since $\bigl|\prod_{i=1}^{s} x_{j_i}\bigr| \leq \|\bl x\|_s^{\,s}$, it suffices to show that the first four
derivatives of $g$ satisfy $|g^{(m)}(u)| \leq c (1 + |u|^{m})$ for all $u$, where here and in what follows 
$c$ is a generic constant that depends on $m$.  The same bound for
$\log(1-\Phi)$ follows by the symmetry $1 - \Phi(u) = \Phi(-u)$.

From \eqref{eqn:psibound} we have, for each $m \geq 1$, 
\begin{equation}
\label{eqn:psi-power}
  \psi(u)^{m} \ \leq (1 + |u|)^{m} \ \leq c\,(1 + |u|^{m}), 
\end{equation}
the last step holding because $(1+|u|)^m = \sum_{j=0}^m \binom{m}{j}|u|^{j}$ and $|u|^{j} \leq 1 + |u|^{m}$ for 
$0 \leq j \leq m$. Differentiating $g' = \psi$ and using $\psi' = -\psi\,(u + \psi)$ repeatedly,
\begin{align*}
  g^{(1)}(u) &= \psi, \\
  g^{(2)}(u) &= -\psi\,(u + \psi), \\
  g^{(3)}(u) &= \psi\,\bigl(u^{2} - 1 + 3u\psi + 2\psi^{2}\bigr), \\
  g^{(4)}(u) &= \psi\,\bigl(-u^{3} + 3u - 6u^{2}\psi + 4\psi - 12u\psi^{2} - 6\psi^{3}\bigr),
\end{align*}
where $\psi = \psi(u)$. Each $g^{(m)}$ is thus a sum of terms of the form $c\,u^{a}\psi^{b}$ with
$a + b \leq m$ and $b \geq 1$. For such a term, $|u|^{a} \leq 1 + |u|^{a} \leq 1 + |u|^{m-b}$, while
\eqref{eqn:psi-power} gives $\psi^{b} \leq  c (1 + |u|^{b})$. Multiplying these terms gives
\[
  |u|^{a}\,\psi(u)^{b} \leq c(1 + |u|^{m-b})(1 + |u|^{b}) \leq c(1 + |u|^{m}),
\]
since $(1 + |u|^{m-b})(1 + |u|^{b}) = 1 + |u|^{m-b} + |u|^{b} + |u|^{m} \leq 4\,(1 + |u|^{m})$. Summing the
finitely many terms of $g^{(m)}$ yields $|g^{(m)}(u)| \leq c(1 + |u|^{m})$ for $m = 1,2,3,4$, which
completes the proof.
\end{proof}

\begin{proof}[Proof of Lemma \ref{lem:delta}]
For the first result, 
 write $\eta_i(\bs\theta) = \alpha - \bl x_i^\top\bs\beta = \bl u_i^\top\bs\theta$ with
$\bl u_i = (1, -\bl x_i^\top)^\top$ and let $Q_n = -\ddot l(\tilde{\bs\theta})/n$. 
Since $\dot l(\tilde{\bs\theta}) = \bl 0$, a third-order Taylor expansion gives 
  $\Delta(\bl s) = - \bl s^\top Q_n \bl s/2 + R(\bl s)$ where the remainder term can be written
\[  R(\bl s) = \frac{1}{6n^{3/2}}\sum_{t,u,v} s_t s_u s_v\, \dddot l^{(t,u,v)}(\bs\theta_*),
\]
with $\bs\theta_* = \tilde{\bs\theta} + \zeta\bl s/\sqrt n$ for some $\zeta \in (0,1)$.
Each third derivative $\dddot l^{(t,u,v)}$ is a sum over $i$ of third derivatives of
$\log\Phi(\eta_i(\bs\theta))$ or $\log(1-\Phi(\eta_i(\bs\theta)))$ in the components of $\bs\theta$. Applying
Lemma~\ref{lem:deriv-bd} with $m = 3$ to each summand, we have
\[
  \big|\dddot l^{(t,u,v)}(\bs\theta_*)\big|
  \ \leq c \sum_{i=1}^n \|\bl x_i\|^{3}\bigl(1 + |\eta_i(\bs\theta_*)|^{3}\bigr),
\]
since $\|\bl x_i\|_s^{\,s} \leq \|\bl x_i\|^{3}$ for $s \leq 3$. To control this uniformly over
$\bl s \in \mathcal S_T$, note that for such $\bl s$,
\[
  |\eta_i(\bs\theta_*)|
  \ \leq |\eta_i(\tilde{\bs\theta})| + \zeta|\bl u_i^\top \bl s|/\sqrt{n}
  \ \leq (1 + \|\bl x_i\|)\Bigl(\|\tilde{\bs\theta}\| + T/\sqrt n\Bigr),
\]
using $\eta_i(\bs\theta) = \bl u_i^\top\bs\theta$, $\|\bl u_i\| \leq 1 + \|\bl x_i\|$ and $\|\bl s\| \leq T$. Hence
for $n \geq T^2$,
\[
  \frac1n\big|\dddot l^{(t,u,v)}(\bs\theta_*)\big|
  \ \leq d \,\bigl(1 + \|\tilde{\bs\theta}\|^{3}\bigr)\cdot\frac1n\sum_{i=1}^n\bigl(\|\bl x_i\|^{3} + \|\bl x_i\|^{6}\bigr),
\]
and by the law of large numbers the average converges to $\Exp{\|\bl x\|^{3} + \|\bl x\|^{6}} < \infty$, while
$\|\tilde{\bs\theta}\| = O_p(1)$. The right-hand side is therefore bounded by a random variable
$B_n = O_p(1)$ not depending on $\bl s$. Since $|s_t s_u s_v| \leq T^3$ on $\mathcal S_T$ and there are
$(p+1)^3$ index triples, $\sup_{\mathcal S_T}|R(\bl s)| \leq B_n T^3/\sqrt n \stackrel{p}{\to} 0$.
Finally, $\sup_{\mathcal S_T}|\bl s^\top(Q_n - I(\bs\theta_0))\bl s/2| \leq T^2\|Q_n - I(\bs\theta_0)\|/2
\stackrel{p}{\to} 0$, and the two bounds combine to give the first result of the lemma. 

For the second result,  we control $\Delta(\bl s)$ on $||\bl s||\geq T$ using the 
concavity of $\Delta(\bl s)$. Fix $\|\bl s\| \geq T$ and put $\bl v = T\bl s/\|\bl s\|$,
a point on the sphere of radius $T$. Since $\bl v$ is a convex combination of $\bl s$ and $\bl 0$ and
$\Delta(\bl 0) = 0$, concavity gives $\Delta(\bl v) \geq (T/\|\bl s\|)\Delta(\bl s)$, or  equivalently  
\begin{equation}
\label{eqn:cone}
  \Delta(\bl s) \leq \frac{\|\bl s\|}{T}\times \sup_{\|\bl v\| = T}\Delta(\bl v)
\end{equation}
for all $\|\bl s \| \geq T$. 
This inequality is deterministic. It remains to bound the supremum on the right.
By the first item of the lemma, the events
$A_n = \{\sup_{\|\bl v\| \leq T}|\Delta(\bl v) +  \bl v^\top I(\bs\theta_0)\bl v/2| \leq 3\lambda_0 T^2/8\}$
have probability converging to one. On $A_n$, since $\bl v^\top I(\bs\theta_0)\bl v \geq \lambda_0 T^2$ when
$\|\bl v\| = T$,
\[
  \sup_{\|\bl v\| = T}\Delta(\bl v)
  \leq -\lambda_0 T^2/2 + 3\lambda_0 T^2/8 = -\lambda_0 T^2/8 ,
\]
and substituting into \eqref{eqn:cone} gives $\Delta(\bl s) \leq -\lambda_0 T\|\bl s\|/8$ for all
$\|\bl s\| \geq T$ on $A_n$, proving the claim.
\end{proof}

\begin{proof}[Proof of Lemma \ref{lem:rho}]
Write $\rho(\bl s) =  ( \pi_\beta( \tilde{\bs\beta} + \bl b /\sqrt{n} )/\pi_\beta(\tilde{\bs\beta}) ) \times 
                       (  w(\tilde{\bs\theta} + \bl s /\sqrt{n} )/w(\tilde{\bs\theta}) )  
                        \equiv \rho_\pi (\bl s) \times  \rho_w(\bl s)$ 
where $\pi_\beta$ is the $N_p(\bs\mu,\Sigma)$ density and $\bl b$ is the $\beta$-block of $\bl s$. 
Then
\begin{equation}
\label{eqn:rho-pi}
  \log\rho_\pi(\bl{s})
  = -\frac{(\tilde{\bs{\beta}} - \bs{\mu})^\top\Sigma^{-1}\bl{b}}{\sqrt n}
    - \frac{\bl{b}^\top\Sigma^{-1}\bl{b}}{2n} .
\end{equation}
Setting $\varepsilon_n = \|\Sigma^{-1}(\tilde{\bs{\beta}} - \bs{\mu})\|/\sqrt n = O_p(n^{-1/2})$ and using $\|\bl{b}\| \leq \|\bl{s}\|$
together with $\bl{b}^\top\Sigma^{-1}\bl{b} \geq 0$ gives the two bounds
\begin{equation}
\label{eqn:rho-pi-bd}
  |\log\rho_\pi(\bl{s})| \leq \varepsilon_n\|\bl{s}\| + \|\bl{s}\|^2/[2 \gamma_0 n]
  \qquad\text{and}\qquad
  \rho_\pi(\bl{s}) \leq e^{\varepsilon_n\|\bl{s}\|} .
\end{equation}
where $\gamma_0$ is the smallest eigenvalue of $\Sigma$.  

Next, writing $\eta_k(\bs\theta) = \bl u_k^\top \bs\theta$, we have 
$w(\bs\theta) = \sum_{k\in N_0} \psi(\eta_k(\bs\theta))$. 
Let $K$ and $K_{1/2}$  be the closed balls of radius $r$ and $r/2$  centered at the true parameter value 
$\bs\theta_0$,
for which $\Pr( \tilde{\bs\theta} \in K)\rightarrow 1$ (and similarly for $K_{1/2}$). The law of large numbers 
gives $n^{-1}\sum_{k\in N_0}\psi(\eta_k(\bs\theta)) \to \bar w(\bs\theta)$ uniformly on the compact set $K$, where
\[
  \bar w(\bs\theta) = \Exp{ \Phi(\alpha_0 - \bl x^\top\bs\beta_0)\,\psi(\bl u^\top\bs\theta)} > 0. 
\]
Being continuous and strictly
positive on the compact set $K$, $\bar w$ satisfies 
$\inf_K \bar w = \min_K \bar w > 0$ deterministically. Now let $D_n$ be the event that 
$w(\bs \theta )$ is close to $\bar w(\bs\theta)$ relative to this lower bound: 
\begin{equation}
\label{eqn:Dn}
  D_n = \Bigl\{ \sup_{\bs\theta \in K}\bigl| n^{-1}w(\bs\theta) - \bar w(\bs\theta)\bigr| \leq \tfrac12\inf_K \bar w \Bigr\}. 
\end{equation}
Note that $\Pr(D_n) \rightarrow 1$.

We use the events $\{D_n\}$ and the mean value theorem to obtain the first result of the lemma. 
The gradient of $\log w(\bs\theta)$ is 
\begin{equation}
\label{eqn:wgrad}
  \nabla \log w(\bs\theta)
  = \frac{\sum_{k\in N_0}\psi'(\eta_k(\bs\theta))\,\bl u_k/n}{w(\bs\theta)/n} .
\end{equation}
The numerator is an average of at most $n$ terms, each having norm at most $\|\bl u_k\| \leq 1 + \|\bl x_k\|$ because
$|\psi'| \leq 1$ by \eqref{eqn:psiprime}.
On $D_n$, for every $\bs\theta \in K$ the denominator of \eqref{eqn:wgrad} satisfies
$w(\bs\theta) \geq n(\inf_K \bar w - \tfrac12\inf_K\bar w) = n\inf_K\bar w/2$, so that
\begin{equation}
\label{eqn:wgrad-bd}
  \sup_{\bs\theta \in K}\|\nabla\log w(\bs\theta)\|
  \leq \frac{(1 + \sum_i \|\bl x_i\|/n)}{\inf_K \bar w/2}
  = O_p(1) .
\end{equation}
Note that \eqref{eqn:wgrad-bd} controls $\nabla\log w$ uniformly over $\bs\theta \in K$.  
To apply it to the 
segment from $\tilde{\bs\theta}$ to $\tilde{\bs\theta} + \bl s/\sqrt n$, we need 
$\tilde{\bs\theta} + \bl s/\sqrt n$ to be in $K$. 
Accordingly, define
\[
  G_n = D_n \cap \{\tilde{\bs\theta} \in K_{1/2}\},
\]
the event on which both the gradient bound \eqref{eqn:wgrad-bd} holds and $\tilde{\bs\theta}$ lies in $K_{1/2}$.
Since $\Pr(D_n) \to 1$ and $\Pr(\tilde{\bs\theta}\in K_{1/2}) \to 1$ we also have $\Pr(G_n) \to 1$.
For $\sqrt n > 2T/r$ and
any $\bl{s}$ with $\|\bl{s}\| \leq T$ and any $\zeta \in [0,1]$,
\[
  \|\tilde{\bs{\theta}} + \zeta\bl{s}/\sqrt n - \bs{\theta}_0\|
  \ \leq \|\tilde{\bs{\theta}} - \bs{\theta}_0\| + \frac{\|\bl{s}\|}{\sqrt n}
  \ \leq \frac r2 + \frac{T}{\sqrt n}
  \ \leq \frac r2 + \frac r2 = r ,
\]
using $\tilde{\bs{\theta}} \in K_{1/2}$ and $T/\sqrt n \leq r/2$. Hence the segment from $\tilde{\bs{\theta}}$
to $\tilde{\bs{\theta}} + \bl{s}/\sqrt n$ lies in $K$, and the mean value theorem together with
\eqref{eqn:wgrad-bd} gives
\[
  \sup_{\|\bl{s}\| \leq T}|\log\rho_w(\bl{s})|
  \leq \frac{T}{\sqrt n}\sup_{\bs{\theta}\in K}\|\nabla\log w(\bs{\theta})\|
  = O_p(T/\sqrt n) .
\]
Combined with the first bound in \eqref{eqn:rho-pi-bd}, which is also $O_p(T/\sqrt n)$ on $\|\bl{s}\| \leq T$,
\[
  \sup_{\|\bl{s}\| \leq T}|\log\rho(\bl{s})|
  \leq \sup_{\|\bl{s}\| \leq T}|\log\rho_\pi(\bl{s})| + \sup_{\|\bl{s}\| \leq T}|\log\rho_w(\bl{s})|
  = O_p(T/\sqrt n) .
\]
Since $\log\rho \stackrel{p}{\rightarrow}0$ uniformly on $\mathcal S_T$ implies $\rho \stackrel{p}{\rightarrow} 1$ on   $\mathcal S_T$, the first result follows.

Finally, to show the second result of the lemma, \eqref{eqn:psibound} gives $\psi(t) \leq 1 + |t|$, so 
\begin{align*} 
  w(\bs\theta)/n  & \leq \sum_{k\in N_0}(1 + \|\bl u_k\|\,\|\bs\theta\|)/n \\
  &   \leq (1 + \|\bs\theta\|)\sum_{k\in N_0}(1 + \|\bl u_k\|)/n \equiv (1+\|\bs\theta\|)\, a_n 
\end{align*}
for every $\bs{\theta}$, where $a_n =O_p(1)$. 
On $G_n$ the point
$\tilde{\bs{\theta}}$ lies in $K_{1/2} \subset K$, so using the lower bound \eqref{eqn:Dn} we have
\[
  \rho_w(\bl{s})
  = \frac{w(\tilde{\bs{\theta}} + \bl{s}/\sqrt n)}{w(\tilde{\bs{\theta}})}
  \leq \frac{a_n(1 + \|\tilde{\bs{\theta}}\| + \|\bl{s}\|/\sqrt n)}{\inf_K\bar w/2}
  \leq C_n(1 + \|\bl{s}\|),
\]
where, using $\|\tilde{\bs\theta}\| = O_p(1)$ and $1 + \|\tilde{\bs{\theta}}\| + \|\bl{s}\|/\sqrt n \leq (1 + \|\tilde{\bs{\theta}}\|)(1 + \|\bl{s}\|)$, 
 $C_n = 2 a_n(1 + \|\tilde{\bs{\theta}}\|)/\inf_K\bar w = O_p(1)$. 
Multiplying by $\rho_\pi(\bl s)$ 
 and applying the second bound in \eqref{eqn:rho-pi-bd}, we have
\[
  \rho(\bl{s}) = \rho_\pi(\bl{s})\,\rho_w(\bl{s}) \leq C_n(1 + \|\bl{s}\|)e^{\varepsilon_n\|\bl{s}\|}
\]
on $G_n$, which proves the second result. 

\end{proof}

\bibliographystyle{plainnat}
\bibliography{refs} 

\end{document}